\documentclass[11pt]{article}
\usepackage{fullpage}

\usepackage{latexsym}
\usepackage{amsmath}
\usepackage{amssymb}
\usepackage{amsfonts}
\usepackage{ifthen}
\usepackage{revsymb}
\usepackage{yfonts}
\usepackage{graphicx}
\usepackage{multirow}
\usepackage{url}

\def\Real{\mathbb{R}}
\def\Complex{\mathbb{C}}

\newcommand{\gammamax}{\gamma_{\mathrm{max}}}

\newcommand{\cancel}[1]{}
\newenvironment{sdp}[2]{
\smallskip
\begin{center}
\begin{tabular}{ll}
#1 & #2\\
subject to
}
{
\end{tabular}
\end{center}
\smallskip
}

\def\01{\{0,1\}}

\newcommand{\eps}{\varepsilon}
\newcommand{\ket}[1]{|#1\rangle}
\newcommand{\bra}[1]{\langle#1|}
\newcommand{\braket}[2]{\langle #1|#2\rangle}

\newcommand{\sketbra}[2]{{\ensuremath{\lvert #1\rangle\!\langle #2\rvert}}}
\newcommand{\lketbra}[2]{{\ensuremath{\left\lvert #1\right\rangle\!\!\left\langle #2\right\rvert}}}
\newcommand{\ketbra}[2]{\if@display\lketbra{#1}{#2}\else\sketbra{#1}{#2}\fi}

\newcommand{\proj}[1]{\ketbra{#1}{#1}}

\newcommand{\gMg}{{g M g^\dag}}

\newcommand{\spr}[2]{\langle #1,#2 \rangle}
\newcommand{\oeps}{\delta}

\newcommand\beq{\begin{equation}}
\newcommand\eeq{\end{equation}}
\newcommand\bea{\begin{eqnarray}}
\newcommand\eea{\end{eqnarray}}
\newcommand{\Tr}{\mbox{\rm Tr}}
\newcommand{\id}{\mathsf{id}}

\newtheorem{theorem}{Theorem}[section]
\newtheorem{lemma}[theorem]{Lemma}
\newtheorem{claim}{Claim}
\newtheorem{corollary}[theorem]{Corollary}
\newtheorem{protocol}{Protocol}
\newtheorem{definition}[theorem]{Definition}

\newenvironment{proof}
{\noindent {\bf Proof. }}
{{\hfill $\Box$}\\
 \smallskip}

\newenvironment{sketch}
{\noindent {\bf Proof (Sketch). }}
{{\hfill $\Box$}\\
 \smallskip}

\newcommand{\ball}{\mathcal{K}^\eps}

\newcommand{\mV}{\mathcal{V}}
\newcommand{\mM}{\mathcal{M}}
\newcommand{\mP}{\mathcal{P}}

\newcommand{\hil}{\mathcal{H}}

\newcommand{\mX}{\mathcal{X}}
\newcommand{\mS}{\mathcal{S}}
\newcommand{\mN}{\mathcal{N}}

\newcommand{\setI}{\mathcal{I}}
\newcommand{\setT}{\mathcal{T}}
\newcommand{\setX}{\mathcal{X}}
\newcommand{\setY}{\mathcal{Y}}
\newcommand{\setF}{\mathcal{F}}

\newcommand{\set}[1]{\{#1\}}

\newcommand{\regX}{X}
\newcommand{\regTheta}{\Theta}
\newcommand{\regE}{E}
\newcommand{\regK}{K}

\newcommand{\bop}{\mathcal{B}}
\newcommand{\pos}{\mathcal{P}}
\newcommand{\states}{\mathcal{S}}

\newcommand{\rank}{\mathrm{rank}}

\newcommand{\assign}{\ensuremath{\kern.5ex\raisebox{.1ex}{\mbox{\rm:}}\kern -.3em =}}

\newcommand{\ol}[1]{\overline{#1}}

\renewcommand{\H}{\operatorname{H}} %  Entropy
\newcommand{\hmin}{\ensuremath{\H_{\infty}}}
\newcommand{\hmine}[2]{\ensuremath{\hmin^{#1}\left(#2\right)}}
\newcommand{\hminee}[1]{\hmine{\varepsilon}{#1}}

\newcommand{\B}{\operatorname{B}} %  figure of merit (in Appendix)
\newcommand{\C}{\operatorname{C}} %  figure of merit (in Appendix)

\newcommand{\syn}{\mathit{syn}} % syndrome information
\newcommand{\pg}{{\rm P}_{\rm guess}}

\begin{document}

\title{Robust Cryptography in the Noisy-Quantum-Storage Model}

\author{\vspace{-1.6cm}}

\author{Christian Schaffner$^1$, Barbara Terhal$^2$ and Stephanie Wehner$^3$\\
\textit{$^1$CWI, P.O.~Box 94079, 1090 GB Amsterdam, The Netherlands}\\
\textit{$^2$IBM, Watson Research Center, P.O.~Box 218, Yorktown Heights, NY, USA}\\
\textit{$^3$Caltech, Institute for Quantum Information, 1200 E California Blvd, Pasadena CA 91125, USA}}
\date{\today}
\maketitle

\begin{abstract}
  It was shown in \cite{prl:noisy} that cryptographic primitives can be
  implemented based on the assumption that quantum storage of qubits
  is noisy. In this work we analyze a protocol for the universal task
  of oblivious transfer that can be implemented using
  quantum-key-distribution (QKD) hardware in the practical setting
  where honest participants are unable to perform noise-free
  operations. We derive trade-offs between the amount
  of storage noise, the amount of noise in the operations performed by
  the honest participants and the security of oblivious transfer which are greatly improved
  compared to the results in \cite{prl:noisy}. As
  an example, we show that for the case of depolarizing noise in
  storage we can obtain secure oblivious transfer as long as the
  quantum bit-error rate of the channel does not exceed 11\%
  {\em and} the noise on the channel is strictly less than the quantum storage noise. This is optimal
for the protocol considered.
  Finally, we show that our analysis easily carries over to
  quantum protocols for secure identification.
\end{abstract}
\vspace{1mm}
\pagestyle{plain}

\section{Introduction}
The noisy-quantum-storage model~\cite{prl:noisy} is based on the
assumption that it is difficult to store quantum states. Based on
current practical and near-future technical limitations, we assume
that any state placed into quantum storage is affected by noise. At
the same time the model assumes that preparation, transmission and
measurement of simple unentangled quantum states can be performed
with much lower levels of noise. The present-day technology of
quantum key distribution with photonic qubits demonstrates this
contrast between a relatively simple technology for
preparation/transmission/measurement versus a limited capability for
quantum storage.

Almost all interesting cryptographic tasks are impossible to realize
without any restrictions on the participating players, neither
classically nor with the help of quantum information, see e.g.
~\cite{lo:insecurity,mayers:trouble,lo&chau:bitcom2,lo&chau:bitcom,mayers:bitcom}.
It is therefore an important task to come up with a cryptographic
model which restricts the capabilities of adversarial players and in
which these tasks become feasible. It turns out that all such
two-party protocols can be based on a simple primitive called 1-2
Oblivious Transfer (1-2 OT)~\cite{kilian:foundingOnOT,GV87}, first
introduced in~\cite{wiesner:conjugate,rabin:ot,even:firstOT}.
In 1-2 OT, the sender Alice starts off with two bit strings $S_0$
and $S_1$, and the receiver Bob holds a choice bit $C$.  The
protocol allows Bob to retrieve $S_C$ in such a way that Alice does
not learn any information about $C$ (thus, Bob cannot simply ask for
$S_C$). At the same time, Alice must be ensured that Bob only learns
$S_C$, and no information about the other string $S_{\ol{C}}$ (thus,
Alice cannot simply send him both $S_0$ and $S_1$). A 1-2 OT
protocol is called unconditionally secure when neither Alice nor Bob
can break these conditions, even when given unlimited resources.

\section{Results}
In this work we focus on the setting where the honest parties are
unable to perform perfect operations and experience errors
themselves, where we analyze individual-storage attacks. These
honest-party errors can be modeled as bit-errors on an effective
channel connecting the honest parties. In unpublished work, we have
shown that for the case of depolarizing noise in storage, security
can be obtained if the actions of the honest parties are noisy
but their error rate does not exceed 2.9\%~\cite{arxiv:noisy}.  This
threshold is too low to be of any practical value. In particular, this result
left open the question whether security can be obtained in a
real-life scenario.

Using a very different analysis, we are now able to show that in the
setting of individual-storage attacks 1-2 oblivious transfer and
secure identification can be achieved in the noisy-storage model
with depolarizing storage noise, as long as the quantum bit-error
rate of the channel does not exceed 11\% and the noise on the
channel is strictly less than the noise during quantum storage. This
is optimal for the protocol considered.

\begin{table}
\begin{center}
\begin{tabular}{|c|c|c|c|}
\hline
& ~~Allowed QBR~~ &~~ Secure 1-2 OT~~ &~~ Secure Identification ~~\\
\hline
\hline
PRL \cite{prl:noisy} & None & Yes & No\\
\hline
Unpublished \cite{arxiv:noisy} & 2.9\% & Yes & No\\
\hline
\multirow{2}{*}{This work} & 11\% & Yes & Yes\\
& (optimal) & & \\
\hline
\end{tabular}
\end{center}
\caption{Summary of previous results and the results in this paper.
The allowed quantum bit-error rate (QBR) is the maximum effective
error-rate on the actions of the honest parties below which we can
prove the security of the cryptographic scheme.}
\end{table}

Our result is of great practical significance, since it paves the
way to achieve security in a real-life implementation. Our main new
Theorems \ref{thm:secure1} and \ref{thm:tradeoffPractical} relate
the security of the 1-2 OT protocol to an uncertainty lower bound on
the conditional Shannon entropy. In order to prove these theorems,
we need to relate the Shannon entropy to the smooth min-entropy and
establish several new properties of the smooth min-entropy, see
Section \ref{sec:propsmooth}.

We evaluate the uncertainty lower bounds on the conditional Shannon
entropy in the practically-interesting case of depolarizing noise
resulting in Theorems \ref{thm:depolarize} and
\ref{thm:depolTradeoff}. From this analysis we obtain the clear-cut
result that, depending on the amount of storage noise, the
adversary's optimal storage attack is to either store the incoming
state as is, or to measure it immediately in one of the two BB84
bases.

\subsection{The Noisy-Quantum-Storage Model and Individual-Storage Attacks}
The noisy-storage model assumes that any quantum state that is placed into
quantum storage is affected by some noise described by a quantum
operation $\mN$.  Practically, noise can arise as a result of
transferring the qubit onto a different physical carrier, for example the transfer of a
photonic qubit onto an atomic ensemble or atomic state. In addition, a
quantum state will undergo noise once it has been transferred into
`storage' if such quantum memory is not $100\%$ reliable.

In principle, one may like to prove security against an adversary
that can perform any operation on the incoming quantum states. Here
however we analyze the restricted case where the adversary Bob
performs \emph{individual-storage} attacks. More precisely, Bob's
actions are of the following form as depicted in
Figure~\ref{figure:individualAttack}.

\begin{enumerate}
\item Bob may choose to (partially) measure (a subset of) his qubits
  immediately upon reception using an error-free {\em product}
  measurement, i.e., when he receives the $j$th qubit, he may apply
  any measurement $\mP_j$ of his choosing.
\item In addition, he can store each incoming qubit, or
  post-measurement state from a prior partial measurement, {\em
  separately} and wait until he gets additional information from Alice (at Step~3
  in Protocol~1). During storage, the $j$th qubit is thereby affected
  by some noise described by a quantum operation $\mN_j$ acting
  independently on each qubit. Note that such quantum operation $\mN_j$ could
  come about from encoding an incoming qubit into an error-correcting code and
  decoding it right before receiving Alice's additional information.
\item Once Bob obtains the additional information he may perform an arbitrary
coherent measurement $\mM$ on his stored qubits and stored classical
data.
\end{enumerate}

We would like to note that we can also derive security if we would
allow Bob to initially perform {\em any}, non-product, destructive
measurement on the incoming qubits. By destructive we mean that
there is no post-measurement quantum data left. The reason is that we
have previously shown in Lemma 2 in \cite{prl:noisy}, that
destructive {\em product} measurements are optimal for Bob if he is not allowed to keep any post-measurement information. 
Hence
this optimality of product measurements reduces such more general
destructive measurement model to our model of individual-storage
attacks. Measurements in present-day technology with single photon
qubits in which photons are detected, are in fact always
destructive, hence our model includes many realistic attacks.
Intuitively, using entangling operations between the incoming qubits
should be of little help in either extracting more information from
these independent, uncorrelated, BB84 qubits or in better preserving
these qubits against noise when the noise is extremely low and more is lost
than gained by measuring some qubits right away and using part of the newly freed space to encode the remaining
qubits. Of course this remains to be proven
(see also Conclusion). What {\em can} help is to entangle an
incoming qubit individually with ancilla qubits in order to store
the incoming qubit in an encoded or other more robust form. This
attack is covered in our model as an effective noisy operation
$\mN_j$ on incoming qubit $j$.

In the following, we use the quantum operation $\mS_i$ to denote the
combined quantum operations of Bob's initial (partial) measurement
and the noise.

\begin{figure}
\begin{center}
 \includegraphics{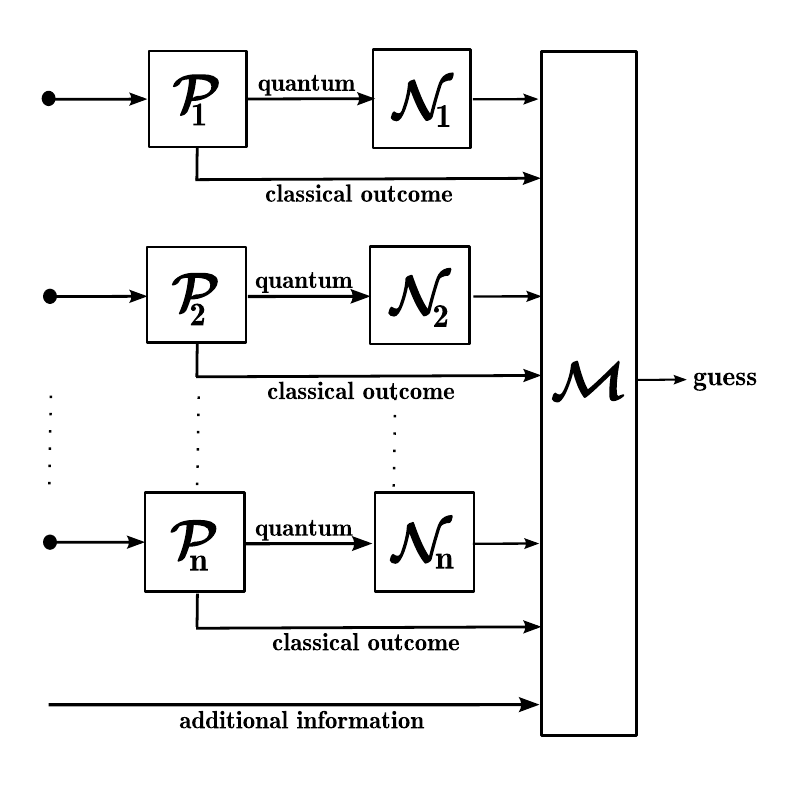}
\caption{Individual-Storage Attacks} \label{figure:individualAttack}
\end{center}
\end{figure}

\subsection{Related work}
Our model is closely related to the bounded-quantum-storage model,
which assumes that the adversary has a limited amount of quantum
memory at his disposal~\cite{serge:bounded,serge:new}.  Within this
`bounded-quantum-storage model' OT can be implemented securely as
long as a dishonest receiver Bob can store at most $n/4-O(1)$ qubits
coherently, where $n$ is the number of qubits transmitted from Alice
to Bob. This approach assumes an explicit limit on the physical
number of qubits (or more precisely, the rank of the adversary's
quantum state). However, at present we do not know of any practical
physical situation which enforces such a limit for quantum
information.
As was pointed out in~\cite{chris:diss,DFSS08journal}, the original
bounded-quantum-storage analysis applies in the case of noise levels
which are so large such that the dishonest player's quantum storage
has an effective noise-free Hilbert space with dimension at most
$2^{n/4}$. The advantage of our model is that we can evaluate the
security parameters of a protocol explicitly in terms of the
strength of the noise, even when the noise rate is very low.

Precursors of the idea of basing cryptographic security on
storage-noise are already present in~\cite{crepeau:practicalOT}, but
no rigorous analysis was carried through in that paper. We
furthermore note that our security proof does not exploit the noise
in the communication channel (which has been done in the classical
setting to achieve cryptographic tasks, see
e.g.~\cite{crepeau:weakenedOT,CMW04,crepeau:efficientOT}), but is
solely based on the fact that the dishonest receiver's quantum
storage is noisy. A model based on classical noisy storage is akin
to the setting of a classical noisy channel, if the operations are
noisy, or the classical bounded-storage model, both of which are
difficult to enforce in practice. Another technical limitation has
been considered in \cite{salvail:physical} where a bit-commitment
scheme was shown secure under the assumption that the dishonest
committer can only measure a limited amount of qubits coherently.
Our analysis differs in that we can in fact allow any coherent
destructive measurement at the end of the protocol.

\subsection{Outline}
In Section~\ref{sec:prelim}, we introduce some notation and the
necessary technical tools. In Section~\ref{sec:12OT}, we define the
security of 1-2 OT, present the protocol and prove its security in
the case when honest players do not experience noise. In
Section~\ref{sec:dep} we then consider the example of depolarizing
storage noise explicitly. The lengthy proof of
Theorem~\ref{thm:depolarize} can be found in
Appendix~\ref{app:mainproof}. In Section~\ref{sec:robust} we show
how to obtain security if the honest players are unable to perform
perfect quantum operations. Finally, we point out in
Section~\ref{sec:others} how our analysis carries over to other
protocols.

\section{Preliminaries} \label{sec:prelim}

We start by introducing the necessary definitions, tools and technical
lemmas that we need in the remainder of this text.

\subsection{Basic Concepts}
We use $\in_R$ to denote the uniform random choice of an element from a set.
We further use $x_{|\setT}$ to denote the string $x=x_1,\ldots,x_n$
restricted to the bits indexed by the set $\setT \subseteq
\{1,\ldots,n\}$. For a binary random variable $C$, we denote by
$\ol{C}$ the bit different from $C$.

Let $\bop(\hil)$ denote the set of all bounded operators on a
finite-dimensional Hilbert space $\hil$. Let $\pos(\hil) \subset
\bop(\hil)$ denote the subset of positive semi-definite Hermitian
operators on $\hil$, and let $\states(\hil) \subset \pos(\hil)$
denote the subset of all quantum states on $\hil$, i.e.~$\rho \in
\states(\hil)$ iff $\rho \in \bop(\hil)$ with $\rho \geq 0$ and
$\Tr(\rho)=1$.
$\Tr_A: \bop(\hil_{AB}) \rightarrow \bop(\hil_B)$ is the partial
trace over system $A$. We denote by $\id_A$ the identity operator on
system $A$.
Let $\ket{0}_+, \ket{1}_+, \ket{0}_\times \assign (\ket{0}_+ +
\ket{1}_+)/\sqrt{2}, \ket{1}_\times \assign (\ket{0}_+ - \ket{1}_+)/\sqrt{2}$ denote
the BB84-states corresponding to the encoding of a classical bit into
the computational or Hadamard basis, respectively.

\paragraph{Classical-Quantum States}
A \emph{cq-state} $\rho_{XE}$ is a state
that is partly classical, partly quantum, and can be written as
$$
\rho_{XE}=\sum_{x \in \setX} P_X(x) \proj{x} \otimes \rho_E^x \, .
$$
Here, $X$ is a classical random variable distributed over the finite
set $\setX$ according to distribution $P_{X}$, $\set{\ket{x}}_{x \in \setX}$
is a set of orthonormal states and the register $E$ is in state $\rho_E^x$ when $X$ takes on value $x$.

\paragraph{Distance measures}
The $L_1$-norm of an operator $A \in  \bop(\hil)$ is defined as
$\|A\|_1 \assign \Tr\sqrt{A^\dagger A}$. The fidelity between two
quantum states $\rho, \sigma$ is defined as $F(\rho,\sigma) \assign
\| \sqrt{\rho} \sqrt{\sigma} \|_1$. For pure states it takes on the
easy form $F(\proj{\phi},\proj{\psi}) = | \braket{\phi}{\psi} |$.
The related quantity $C(\rho,\sigma) \assign
\sqrt{1-F^2(\rho,\sigma)}$ is a convenient distance measure on
normalized states~\cite{GLN05}. It is invariant under purifications
and equals the trace distance for pure states,
i.e.~$C(\proj{\psi},\proj{\phi}) = \sqrt{1-|\braket{\psi}{\phi}|^2}
= \frac12 \| \proj{\psi} - \proj{\phi} \|_1$.

\paragraph{Non-uniformity}
We can say that a quantum adversary has little information about $X$
if the distribution $P_X$ given his quantum state is close to
uniform. Formally, this distance is quantified by the {\em
non-uniformity} of $X$ given $\rho_E = \sum_x P_X(x) \rho_E^x$
defined as
\begin{equation}
d(X|E) := \frac{1}{2}\left\|\,\id_X/|\mX| \otimes \rho_E-\sum_{x}P_{X}(x) \proj{x} \otimes \rho_{E}^x\,\right\|_{1} \, .
\end{equation}
Intuitively,
$d(X|E) \leq \eps$ means that the distribution of $X$ is $\eps$-close to
uniform even given $\rho_E$, i.e., $\rho_E$ gives hardly any
information about $X$.
A simple property of the non-uniformity which follows from its
definition is that it does not change given independent information. Formally,
\begin{equation}
d(X|E,D) =d(X|E)
\label{eq:indep}
\end{equation}
for any cqq-state of the form $\rho_{XED}=\rho_{XE} \otimes \rho_D$.

\subsection{Entropic Quantities}
Throughout this paper we use a number of entropic quantities. The
\emph{binary-entropy} function is defined as $h(p) \assign - p \log p
- (1-p) \log (1-p)$, where $\log$ denotes the logarithm base 2
throughout this paper.  The \emph{von Neumann entropy} of a quantum
state $\rho$ is given by
$$
\H(\rho) \assign - \Tr(\rho \log \rho) \, .
$$
For a bipartite state $\rho_{AB} \in \mS(\hil_{AB})$,
we use the shorthand
$$
\H(A|B) \assign \H(\rho_{AB}) - \H(\rho_{B})
$$
to denote the \emph{conditional von Neumann entropy} of
the state $\rho_{AB}$ given the
quantum state $\rho_{B} = \Tr_A(\rho_{AB}) \in
\mS(\hil_B)$. 
Of particular importance to us are the following quantities
introduced by Renner~\cite{renato:diss}. Let $\rho_{AB} \in
\mS(\hil_{AB})$.  Then the \emph{conditional min-entropy of
$\rho_{AB}$ relative to $B$} is defined by the following
semi-definite program
$$
\hmin(A|B)_\rho \assign -\log \min_{\substack{\sigma_B \in
\pos(\hil_B) \\ \rho_{AB} \leq \id_A \otimes \sigma_B}}
\Tr(\sigma_B) \, .
$$
For a cq-state $\rho_{XE}$ one can show~\cite{KRS09} that the
conditional min-entropy is the (negative logarithm of the) guessing
probability
\footnote{Such an ``operational meaning'' of conditional
  min-entropy can also be formulated for general
  qq-states~\cite{KRS09}.}
\begin{equation} \label{eq:duality}
\hmin(X|E)_{\rho} = -\log \pg(X|E)_{\rho} \, ,
\end{equation}
where $\pg(X|E)_{\rho}$ is defined as the maximum success
probability of guessing $X$ by measuring the $E$-register of
$\rho_{XE}$. Formally, for any (not necessarily normalized) cq-state
$\rho_{XE}$, we define
\[ \pg(X|E)_\rho \assign \sup_{\set{M_x}} \sum_x P_X(x) \Tr(M_x \rho_E^x) \, ,
\]
where the supremum ranges over all positive-operator valued
measurements (POVMs) with measurement elements $\set{M_x}_{x \in
  \mX}$, i.e.~$M_x \geq 0$ and $\sum_x M_x = \id_E$. If all information
  in $E$ is classical, we recover the fact that the classical min-entropy
  is the negative logarithm of the average maximum guessing probability.

In our proofs we also need smooth versions of these entropic
quantities. The idea is to no longer consider the min-entropy of a
fixed state $\rho_{AB}$, but take the supremum over the min-entropy
of states $\hat{\rho}_{AB}$ which are close to $\rho_{AB}$, and
which may have considerably larger min-entropy. In a cryptographic
setting, we are often not interested in the min-entropy of a
concrete state $\rho_{AB}$, but in the maximal min-entropy we can
get from states in the neighborhood of $\rho_{AB}$, i.e.~deviating
only slightly from the real situation $\rho_{AB}$. These smooth
quantities have some nice properties which are needed in our
security proof. For $\eps \geq 0$, the \emph{$\eps$-smooth
  min-entropy of $\rho_{AB}$} is given by
$$
\hminee{A|B}_\rho \assign \sup_{\hat{\rho}_{AB} \in
  \ball(\rho_{AB})} \hmin(A|B)_{\hat{\rho}} \, ,
$$
where $\ball(\rho_{AB}) \assign \{\hat{\rho}_{AB} \in \pos(\hil_{AB})\mid C(\rho_{AB},\hat{\rho}_{AB}) \leq \eps \mbox{ and } \Tr(\hat{\rho}_{AB}) \leq 1)\}$. If the quantum states $\rho$ are clear from the context, we drop the subscript of the entropies.

\subsubsection{Properties of The Conditional Smooth Min-Entropy}
\label{sec:propsmooth} In our security analysis we make use of the
following properties of smooth min-entropy. First, we need the chain
rule whose simple proof can be found in Appendix~\ref{app:chain}:
\begin{lemma}[Chain Rule] \label{lem:chain}
For any ccq-state $\rho_{XYE} \in \states(\hil_{XYE})$
and for all $\eps \geq 0$, it holds that
$$
\hmine{\eps}{X | Y E } \geq \hminee{XY | E} - \log|\setY|,
$$
where $|\setY|$ is the alphabet size of the random variable $Y$.
\end{lemma}

Secondly, we prove the additivity of the smooth conditional
min-entropy (see Appendix~\ref{app:additivity}):
\begin{lemma}[Additivity] \label{lem:additivity}
Let $\rho_{AB}$ and $\rho_{A'B'}$ be two independent qq-states. For
$\eps \geq 0$, it holds that
\[ \hmine{\eps^2}{AA'|BB'}_{\rho} \leq \hmine{\eps}{A|B}+\hmine{\eps}{A'|B'}\, .
\]
\end{lemma}

Thirdly, adding a classical register can only increase the smooth
min-entropy (see Appendix~\ref{app:monotonicity}):
\begin{lemma}[Monotonicity] \label{lem:monotone}
For a ccq-state $\rho_{XYE}$ and for all $\eps \geq 0$, it holds that
\begin{equation} \nonumber
\hminee{XY | E}  \geq \hminee{Y | E} \, .
\end{equation}
\end{lemma}

At last, we deduce a lower bound on the conditional smooth
min-entropy of product states. The following theorem is a
straightforward generalization of Theorem~7 in~\cite{TCR08} (see
also~\cite[Theorem~3.3.6]{renato:diss}) to the case where the states
are independently, but not necessarily identically distributed. The
theorem states that for a large number of independent states, the
conditional smooth min-entropy can be lower-bounded by the
conditional Shannon entropy. We note that it is a common feature of
equipartition theorems for classical or quantum information that the
assumption of i.i.d.~sources can be replaced by the weaker
assumption of non-i.i.d.~but independent sources (see
Appendix~\ref{app:aequi} for the proof).

\begin{theorem}[adapted from~\cite{TCR08}]\label{thm:aequi}
For $i=1,\ldots,n$, let $\rho_{i} \in \mS(\hil_{AB})$ be density operators.
Then, for any $\eps > 0$,
\begin{equation*}
\hminee{A^n|B^n}_{\bigotimes_{i=1}^n \rho_{i}} \geq \sum_{i=1}^n \left[ \H(A_i|B_i)_{\rho_i} \right] -
  \delta(\eps,\gamma)\sqrt{n}
\, ,
\end{equation*}
where, for $n \geq \frac{8}{5} \log \frac{2}{\eps^2}$, the error is given by
\[ \delta(\eps,\gamma) \assign 4 \log \gamma \sqrt{\log \frac{2}{\eps^2}}
\]
and the single-system entropy contribution by
\[ \gamma \leq 2 \max_i \sqrt{\rank(\rho_{A_i})} + 1 \, .
\]
\end{theorem}

For the case of independent cq-states in Hilbert spaces with the same dimensions, we obtain
\begin{corollary} \label{cor:AEP} % asymptotic equipartition
  For $i=1,\ldots,n$, let $\rho_{X_i B_i}$ be cq-states over (copies
  of) the same space $\hil_X \otimes \hil_B$. Then for every $\eps>0$ and $n \geq \frac{8}{5} \log \frac{2}{\eps^2}$,
\begin{equation} \label{eq:defRow}
\hminee{ X^n| B^n}_{\bigotimes_{i=1}^n \rho_{X_i B_i}} \geq \sum_{i=1}^n
\H(X| B)_{\rho_{X_i B_i}} - \delta n \, ,
\end{equation}
where $\delta \assign \sqrt{\frac{\log(2/\eps^2)}{n}} 4 \log(2\sqrt{\dim{\hil_{X}}}+1)$.
\end{corollary}

We use the properties of the smooth min-entropy to prove the
following two lemmas. These lemmas show that the (smooth)
min-entropy of two independent strings can be split.

\begin{lemma} \label{lemma:ESLnew} Let $\varepsilon \geq 0$, and let
  $\rho_{X_0 E_0},\rho_{X_1 E_1}$ be two independent cq-states with
$$  
\mbox{$\hmin^{\eps^2}(X_0 X_1 | E_0 E_1) \geq \alpha$}\ .
$$  
Additionally, let $S_0, S_1$ be
  classical random variables distributed over $\{0,1\}^\ell$.  Then,
  there exists a random variable $D \in \set{0,1}$ such that
  \mbox{$\hmin^{\eps}(X_{\ol{D}} D S_D|E_0 E_1) \geq \alpha/2$}.
\end{lemma}
\begin{proof}
  From the additivity of smooth min-entropy
  (Lemma~\ref{lem:additivity}) it follows that we can split the
  min-entropy as
\begin{equation*}
\hmine{\eps}{X_0 | E_0} + \hmine{\eps}{X_1 | E_1} \geq
\hmine{\eps^2}{X_0 X_1 | E_0 E_1} \geq \alpha \, ,
\end{equation*}
and therefore, there exists $D \in \set{0,1}$ such that
$$
\hmine{\eps}{X_{\ol{D}} D S_D| E_0 E_1} \geq \alpha/2 \, ,
$$
where we used the monotonicity of smooth min-entropy
(Lemma~\ref{lem:monotone}).
\end{proof}

\begin{lemma} \label{lemma:ESLident} Let $\eps \geq 0$. Let $\rho_{X
    E} = \bigotimes_{i=0}^{m-1} \rho_{X_i E_i}$ be a cq-state consisting
  of $m$ independent cq-substates such that $\hmine{\eps^2}{X_i X_j|E}
  \geq \alpha$ for all $i \neq j$. Then there exists a random variable
  $V$ over $\set{1,\ldots,m}$ such that for any $v \in
  \set{1,\ldots,m}$ with $P[V\!\neq\!v] > 0$
$$
\hmine{\eps}{X_v|E_v V,V\!\neq\!v} \geq \alpha/2 - \log(m)\, .
$$
\end{lemma}
\begin{proof}
  Let $V \in \set{0, \ldots, m-1}$ be the index which achieves
  the minimum of $\hmine{\eps}{X_i | E_i}$, i.e.
  $\hmine{\eps}{X_V | E_V} = \min_i
  \hmine{\eps}{X_i | E_i}$. By the additivity of smooth
  min-entropy (Lemma~\ref{lem:additivity}), we have for all $v \neq V$,
$$\alpha \leq \hmine{\eps^2}{X_v X_V| E} \leq \hmine{\eps}{X_v | E_v}
+ \hmine{\eps}{X_V |  E_V} \, .$$
It follows that $\hmine{\eps}{X_v | E_v, V \neq
  v} \geq \alpha /2$. The chain rule (Lemma~\ref{lem:chain}) then
leads to the claim.
\end{proof}

\subsection{Tools}

We also require the following technical results. This lemma is well-known, see \cite{AS00} or \cite{MR95} for a proof.
\begin{lemma}[Chernoff's inequality] \label{lem:chernoff}
  Let $X_1,\ldots,X_n$ be identically and independently distributed
  random variables with Bernoulli distribution, i.e.~$X_i=1$ with
  probability $p$ and $X_i=0$ with probability $1-p$. Then $S \assign
  \sum_{i=1}^n X_i$ has a binomial distribution with parameters $(n,p)$
  and it holds that
\[ \Pr\left[ \: |S - pn| > \eps n \: \right] \leq 2 e^{- 2 \eps^2 n} \, .
\]
\end{lemma}

\paragraph{Privacy Amplification}
The OT protocol makes use of two-universal hash functions.  These hash
functions are used for privacy amplification similar as in quantum key
distribution. A class $\setF$ of functions $f: \01^n \rightarrow
\01^\ell$ is called two-universal, if for all $x\neq y \in \01^n$ and
$f \in \setF$ chosen uniformly at random from $\setF$, we have
$\Pr[f(x) = f(y)] \leq 2^{-\ell}$~\cite{CarWeg79}.
The following theorem expresses how the application of hash functions
can increase the privacy of a random variable X given a quantum
adversary holding $\rho_E$, the function $F$ and a classical random
variable $U$:

\begin{theorem}[\cite{renato:diss,serge:new}] \label{thm:PA} Let
  $\setF$ be a class of two-universal hash functions from $\01^n$ to
  $\01^\ell$.  Let $F$ be a random variable that is uniformly and
  independently distributed over $\setF$, and let $\rho_{XUE}$ be a
  ccq-state. Then, for any $\eps \geq 0$,
$$
d(F(X)|F,U,E) \leq 2^{-\frac{1}{2}\left(\hminee{X|UE} - \ell
    \right)-1}+ \eps \, .
$$
\end{theorem}

\section{1-2 Oblivious Transfer} \label{sec:12OT}

\subsection{Security Definition and Protocol}

In this section we prove the security of a randomized version of 1-2
OT (Theorem \ref{thm:secure1}) from which we can easily obtain 1-2
OT. In such a randomized 1-2 OT protocol, Alice does not input two
strings herself, but instead receives two strings $S_0$, $S_1 \in
\01^\ell$ chosen uniformly at random. Randomized OT (ROT) can easily
be converted into OT. After the ROT protocol is completed, Alice uses
her strings $S_0,S_1$ obtained from ROT as one-time pads to encrypt
her original inputs $\hat{S_0}$ and $\hat{S_1}$, i.e.~she sends an
additional classical message consisting of $\hat{S_0} \oplus S_0$ and
$\hat{S_1} \oplus S_1$ to Bob. Bob can retrieve the message of his
choice by computing $S_C \oplus (\hat{S}_C \oplus S_C) =
\hat{S}_C$. He stays completely ignorant about the other message
$\hat{S}_{\ol{C}}$ since he is ignorant about $S_{\ol{C}}$. The
security of a quantum protocol implementing ROT is formally defined in
\cite{serge:new} and justified in~\cite{FS09} (see
also~\cite{WW08:compose}).

\begin{definition} \label{def:ROT}
An $\eps$-secure 1-2 $\mbox{ROT}^\ell$ is a protocol between Alice
and Bob, where Bob has input $C \in \01$, and Alice has no input.
\begin{itemize}
\item (Correctness) If both parties are honest, then for any
  distribution of Bob's input $C$, Alice gets outputs $S_0,S_1 \in
  \01^\ell$ which are $\eps$-close to uniform and independent of $C$
  and Bob learns $Y = S_C$ except with probability $\eps$.
\item (Security against dishonest Alice) If Bob is honest and obtains output $Y$,
  then for any cheating strategy of Alice resulting in her state
  $\rho_A$, there exist random variables $S'_0$ and $S'_1$ such that
  $\Pr[Y=S'_C] \geq 1- \eps$ and $C$ is independent of $S'_0$,$S'_1$
  and $\rho_A$\footnote{Existence of the random variables
    $S'_0,S'_1$ has to be understood as follows: given the cq-state
    $\rho_{Y \! A}$ of honest Bob and dishonest Alice, there exists a
    cccq-state $\rho_{Y S'_0 S'_1 A}$ such that tracing out the
    registers of $S'_0,S'_1$ yields the original state $\rho_{Y A}$
    and the stated properties hold.}.
\item (Security against dishonest Bob) If Alice is honest, then for any cheating strategy of Bob resulting in his state $\rho_B$,
there exists a random variable $D \in \01$ such that $d(S_{\ol{D}}|S_{D}D\rho_B) \leq \eps$.
\end{itemize}
\end{definition}

For convenience, we choose $\{+,\times\}$ instead of $\01$ as domain
of Bob's choice bit $C$. We consider the same protocol for ROT as
in~\cite{serge:new}.
\begin{protocol}[\cite{serge:new}]1-2 $\mbox{ROT}^\ell$ \label{prot:nonoise}
\begin{enumerate}
\item Alice picks $X \in_R \01^n$ and $\Theta \in_R \{+,\times\}^n$.
  Let $\setI_b = \{i\mid \Theta_i = b\}$ for $b \in \{+,\times\}$. At
  time $t=0$, she sends $\ket{X_1}_{\Theta_1},\ldots,\ket{X_n}_{\Theta_n}$ to Bob.
\item Bob measures all qubits in the basis corresponding to his choice
  bit $C \in \{+,\times\}$.  He obtains outcome $X' \in \01^n$.
\item Alice picks two hash functions $F_+,F_\times \in_R
  \setF$, where $\setF$ is a class of two-universal hash functions.
   At the reveal time $t=T_{\rm rev}$, she sends $\setI_+$,$\setI_\times$, $F_+$,$F_\times$ to Bob.
Alice outputs $S_+ = F_+(X_{|\setI_+})$ and $S_{\times}
  = F_\times(X_{|\setI_\times})$ \footnote{If $X_{|\setI_b}$ is less
    than $n$ bits long Alice pads the string $X_{|\setI_b}$ with 0's
    to get an $n$ bit-string in order to apply the hash function to
    $n$ bits.}.
\item Bob outputs $S_C = F_C(X'_{|\setI_C})$.
\end{enumerate}
\end{protocol}

\subsection{Security Analysis}
We show in this section that Protocol~\ref{prot:nonoise} is secure
according to Definition~\ref{def:ROT}, in case the dishonest receiver
is restricted to individual-storage attacks. 
\begin{figure}
\begin{center}
 \includegraphics{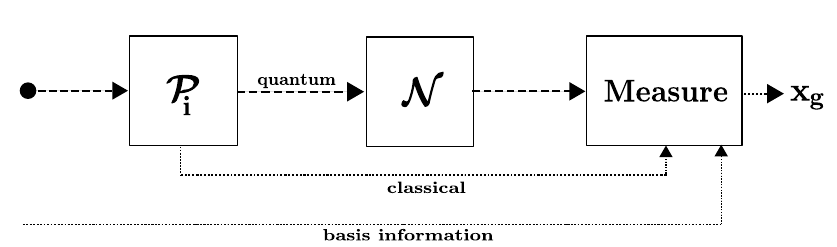}
  \caption{Bob performs a partial measurement $\mP_i$, followed by
    noise $\mN$, and outputs a guess bit $x_g$ depending on his
    classical measurement outcome, the remaining quantum state, and
    the additional basis
    information. \label{figure:depolModel}}
\end{center}
\end{figure}

\paragraph{Correctness}
First of all, note that it is clear that the protocol fulfills its task correctly.  Bob can
determine the string $X_{|\setI_C}$ (except with negligible
probability $2^{-n}$ the set ${\cal I}_C$ is non-empty) and hence
obtains $S_C$.  Alice's outputs $S_+,S_\times$ are
  perfectly independent of each other and of $C$.

\paragraph{Security against Dishonest Alice}
Security holds in the same way as
shown in~\cite{serge:new}. As the protocol is non-interactive, Alice
never receives any information from Bob at all, and Alice's input
strings can be extracted by letting her interact with an unbounded
receiver.

\paragraph{Security against Dishonest Bob}

Proving that the protocol is secure against Bob requires more work.
Our goal is to show that there
exists a $D \in \{+,\times\}$ such that Bob is completely ignorant
about $S_{\ol{D}}$.

Recall that in round $i$, honest Alice picks $X_i \in_R \set{0,1}$
and $\Theta_i \in_R \set{+,\times}$ and sends $\ket{X_i}_{\Theta_i}$
to dishonest Bob. Bob can subsequently do a partial measurement to
obtain the classical outcome $K_i$ and store the remaining quantum
state in register $E_i$ which is then subject to noise (see
Figure~\ref{figure:depolModel}). It is important to note that Bob's
initial partial measurement does not depend on the basis information
$\Theta$. Since we are modeling individual-storage attacks, the
overall state (as viewed by Bob) for Alice and Bob right before time
$T_{\rm rev}$ is of the form
$$
\rho_{X \Theta K E } = \bigotimes_{i=1}^n \rho_{X_i \Theta_i K_i E_i}
\, ,
$$
with
\begin{equation}\label{eq:bobState}
\rho_{X_i \Theta_i K_i E_i} = \frac{1}{4} \sum_{x_i, \theta_i,k_i } p_{k_i | x_i \theta_i}
\underbrace{\proj{x_i}}_{\regX_i} \otimes \underbrace{\proj{\theta_i}}_{\regTheta_i} \otimes \underbrace{\proj{k_i}}_{\regK_i} \otimes
\underbrace{\mN_i \left( \rho_{x_i \theta_i}^{k_i}  \right)}_{\regE_i}
\, ,
\end{equation}
where we use $\regX_i$ to denote Alice's system corresponding to her
choice of bit $x_i$, $\regTheta_i$ for the system corresponding to her
choice of basis $\theta_i$, and $\regK_i$ and $\regE_i$ for Bob's
systems corresponding to the classical outcome $k_i$ (with probability
$p_{k_i|x_i \theta_i}$) of his partial measurement and his remaining
quantum system respectively.

It is clear that a dishonest receiver will have some uncertainty about
the bit $X_i$, given that he either measured the register $E$ without the correct basis information and that
storage noise occurred on the post-measurement quantum state. To formalize this uncertainty, let us call
$t$ an {\em uncertainty lower bound} on the conditional
Shannon entropy if, for all $i=1,\ldots,n$, we have
\begin{equation} \label{eq:vNlower}
\H(X_i|\Theta_i K_i E_i) = \H(\rho_{X_i \Theta_i K_i E_i}) - \H(\rho_{\Theta_i K_i E_i}) \geq t \, .
\end{equation}
The parameter $t$ thereby depends on the specific kind of noise in the quantum storage.
In Section~\ref{sec:dep}, we evaluate the uncertainty lower-bound $t$ for the case of depolarizing noise.

The following theorem shows that as long as $\ell
\lesssim tn/4$, the protocol is secure except with probability $\eps$.
Since we are performing $1$-out-of-$2$ oblivious transfer of $\ell$-bit strings, $\ell$ corresponds
to the ``amount'' of oblivious transfer we can perform for a given security parameter
$\eps$ and number of qubits $n$. In QKD, $\ell$ corresponds to the length of the key
generated.

\begin{theorem}
Protocol~\ref{prot:nonoise} is $2\eps$-secure against a dishonest
receiver Bob according to Definition~\ref{def:ROT}, if $n \geq \frac{8}{5} \log \frac{2}{\eps^4}$ and
$$
\ell \leq \frac{1}{4} \left(t - \delta \right)n + \frac{1}{2} -
\log\left(\frac{1}{\eps}\right) \, ,
$$
where $\delta = 8\sqrt{\log(2/\eps^4)/n}$, and $t$ is the
uncertainty lower bound on the conditional Shannon entropy
fulfilling Eq.~\eqref{eq:vNlower}. \label{thm:secure1}
\end{theorem}

\begin{proof}
  We need to show the existence of a binary random variable $D$ such
  that $S_{\bar{D}}$ is $\eps$-close to uniform from Bob's point of
  view. As noted above, the overall state of Alice and Bob before time
  $T_{\rm rev}$ has a product form. After time $T_{\rm rev}$, dishonest Bob holds the
  classical registers $\Theta,K$, the quantum register $E$ as well as
  classical information about the hash functions $F_+,F_\times$. To
  prove security, we first lower-bound Bob's uncertainty about $X$ in
  terms of min-entropy, use Lemma~\ref{lemma:ESLnew} to obtain
  $D$ and then apply the privacy amplification theorem.

  First of all, we know from Corollary~\ref{cor:AEP} that
  the smooth min-entropy of an $n$-fold tensor state is roughly equal
  to $n$ times the von Neumann entropy of its substates.  Hence,
  applying Corollary~\ref{cor:AEP} to our setting with $B_i \assign
  \Theta_i K_i E_i$ and $\log(2\sqrt{\dim\hil_{X_i}}+1)=\log(2\sqrt{2}+1) \leq 2$ we obtain for $n \geq \frac{8}{5} \log \frac{2}{\eps^4}$ that
\begin{equation*}
  \hmine{\eps^2}{X | \Theta K E} \geq \sum_{i=1}^n \H(X_i | \Theta_i K_i E_i)  - \delta n \geq (t-\delta)n \, ,
\end{equation*}
with $\delta = 8\sqrt{\log(2/\eps^4)/n}$. We used Equation~(\ref{eq:defRow}) in the first inequality and the last
follows by Definition~(\ref{eq:vNlower}) of the uncertainty bound $t$.

For ease of notation, we use $X_+$ and $X_\times$ to denote
$X_{|\setI_+}$ and $X_{|\setI_\times}$, the sequences of bits $X_i$
where $\Theta_i=+$ and $\Theta_i=\times$, respectively. From
$\hmine{\eps^2}{X_+ X_\times | \Theta K E} \geq (t-\delta)n$ and
Lemma~\ref{lemma:ESLnew} it follows that $D \in \set{+,\times}$ exists
such that
$$
\hmine{\eps}{X_{\ol{D}} D S_D| \Theta K E} \geq
(t - \delta)\frac{n}{2} \, .
$$
The rest of the security proof is analogous to the derivation
in~\cite{serge:new}: It follows from the chain rule
(Lemma~\ref{lem:chain}) and the monotonicity~(Lemma~\ref{lem:monotone})
of the smooth min-entropy that
\begin{align*}
\hmine{\eps}{X_{\ol{D}}|\Theta D S_D K E}
&\geq \hmine{\eps}{X_{\ol{D}}D S_D|\Theta K E} - (\ell + 1) \\
&\geq (t - \delta)\frac{n}{2} - 1 - \ell.
\end{align*}
The privacy amplification Theorem~\ref{thm:PA} yields
\begin{equation}\label{eq:explicitTradeoff}
d(F_{\ol{D}}(X_{\ol{D}}) \mid \Theta F_D D S_D K E) \leq
2^{-\frac12((t - \delta)\frac{n}{2} + 1 - 2\ell) }+ \eps
\,
\end{equation}
which is smaller than $2 \eps$ as long as
$$
(t - \delta)\frac{n}{4} + \frac{1}{2} -
\ell \geq \log\left(\frac{1}{\eps}\right) \, .
$$
from which our claim follows.
\end{proof}
We note that one can improve on the extractable length $\ell$ by using
a quantum version of Wullschleger's distributed-privacy-amplification
theorem~\cite{Wullsc07}. Since this technique is specific to
oblivious transfer and does not apply to our extension to the case of
secure identification, we do not go into the details here.

\section{Example: Depolarizing Noise} \label{sec:dep}

In this section, we consider the case when Bob's storage is affected by
depolarizing noise as described by the quantum operation
\begin{equation} \label{eq:depolChannel}
\mN(\rho) = r \rho + (1-r) \frac{\id}{2}.
\end{equation}
Depolarization noise will leave the input state $\rho$ intact with
probability $r$, but replace it with the completely mixed state with
probability $1-r$.  In order to give explicit security parameters
for this setting, our goal is to prove an uncertainty bound $t$ for
the conditional von Neumann entropy $\H(X_i|\Theta_i K_i E_i)$ as in
Eq.~(\ref{eq:vNlower}). Exploiting the symmetries in the
setting, we derive in Appendix~\ref{app:mainproof} the following
result. We drop the index $i$ in this Theorem.

\begin{theorem} \label{thm:depolarize} Let $\mN$ be the depolarizing
quantum operation given by Eq.~\eqref{eq:depolChannel} and let $\H(X
| \Theta
  K E)$ be the conditional von Neumann entropy of one qubit. Then
$$
\H(X | \Theta K E) \geq \left\{
\begin{array}{ll}
h(\frac{1+r}{2}) & \mbox{ for } r \geq \hat{r} \, ,  \\
1/2  & \mbox{ for } r < \hat{r} \, ,
\end{array}\right.
$$
where $\hat{r} \assign 2 h^{-1}(1/2)-1 \approx 0.7798$.
\end{theorem}
Our result shows that when the probability of retaining the input
state $\rho$ is small, $r < 0.7798$, the best attack for Bob is to
measure everything right away in the computational basis.  For this
measurement, we have $\H(X|\Theta K E) \geq 1/2$.  If the depolarizing
rate is low, i.e.~$r \geq 0.7798$, our result says that the best
strategy for Bob is to simply store the qubit as is.

Our result may seem contradictory to our previous error trade-off
obtained in~\cite{arxiv:noisy}, where Bob's best strategy was to
either store the qubit as is or measure it in the \emph{Breidbart
  basis} depending on the amount of depolarizing noise.  Note,
however, that the quantity we optimize in this work is the \emph{von
  Neumann entropy} and not the guessing probability considered
in~\cite{arxiv:noisy}. This phenomenon is similar to the setting of
QKD, where Eve's strategy that optimizes her guessing probability is
different from the one that optimizes the entropy~\cite{GRTZ:qkd_review}.
In general, the von Neumann entropy is larger than the min-entropy
(which corresponds to the guessing probability).
Corollary~\ref{cor:AEP} provides the explanation why the von Neumann
entropy is the relevant quantity in the setting of
individual-storage attacks.

\section{Robust Oblivious Transfer}
\label{sec:robust} In a practical setting, honest Alice and honest
Bob are not able to perform perfect quantum operations or transmit
qubits through a noiseless channel. We must therefore modify the ROT
protocol to make it robust against noise for the honest parties. The
protocol we consider is a small modification of the protocol
considered in~\cite{chris:diss}. The idea is to let Alice send
additional error-correcting information which can help honest Bob to
retrieve $S_C$ as desired.  The main difficulty in the analysis of
the extended protocol is the fact that we have to assume a
worst-case scenario: If Bob is dishonest, we give him access to a
perfect noise-free quantum channel with Alice and he only
experiences noise during storage.

We can divide the noise on the channel into two categories, which we
consider separately: First, we consider erasure noise (in practice
corresponding to photon loss) during preparation, transmission and
measurement of the qubits by the honest parties. Let $1-p_{\rm erase}$
be the total probability for an honest Bob to measure and detect a
photon in the $\{+,\times\}$-basis given that an honest Alice prepares
a weak pulse in her lab and sends it to him. The probability $p_{\rm
  erase}$ is determined, among other things, by the mean photon number in the
pulse, the loss on the channel and the quantum efficiency of the
detector. In our protocol we assume that the erasure rate
$p_{\rm erase}$ is {\em independent} for every pulse and independent
of whether qubits were encoded or measured in the $+$- or
$\times$-basis whenever Bob is honest. 
This assumption is necessary to guarantee the
correctness and the security against a \emph{cheating Alice} only.
Fortunately, this assumption is well matched with the possible
physical implementations of the protocol.

Any other noise source during preparation, transmission and
measurement can be characterized as an effective classical noisy
channel resulting in the output bits $X'$ that Bob obtains at
Step~\ref{step:reception} of Protocol~\ref{prot:practical}.  For
simplicity, we model this compound noise source as a classical
binary symmetric channel acting independently on each bit of $X$.
Typical noise sources for polarization-encoded qubits are
depolarization during transmission, dark counts in Bob's detector
and misaligned polarizing beam-splitters. Let the effective
bit-error probability, called the quantum bit-error rate in quantum
key distribution, of this binary symmetric channel be $p_{\rm error}
< 1/2$.

\subsection{Protocol}
In this section we present the modified version of the ROT protocol.
Before engaging in the actual protocol, Alice and Bob agree on a
small enough security-error probability $\eps>0$ that they are
willing to tolerate. In addition, they determine the system
parameters $p_{\rm erase}$ and $p_{\rm error}$ similarly to Step~1
of the protocol in~\cite{crepeau:practicalOT}. Furthermore, they
agree on a family $\set{C_n}$ of linear error-correcting codes of
length $n$ capable of efficiently correcting $n \cdot p_{\rm error}$
errors~\cite{crepeau:efficientOT}. For any string $x \in
\set{0,1}^n$, error-correction is done by sending the syndrome
information $\syn(x)$ to Bob from which he can correctly recover $x$
if he holds an output $x' \in \set{0,1}^n$ obtained by flipping each
bit of $x$ independently with probability $p_{\rm error}$. It is
known that for large enough $n$, the code $C_n$ can be chosen such
that its rate is arbitrarily close to $1-h(p_{\rm error})$ and the
syndrome length (the number of parity check bits) is asymptotically
bounded by $|\syn(x)| < h(p_{\rm error})
n$~\cite{crepeau:efficientOT}.  We assume that the players have
synchronized clocks. In each time slot, Alice sends one qubit to
Bob.

\begin{protocol}Robust 1-2 $\mbox{ROT}^\ell(C,T,\eps)$ \label{prot:practical}
\begin{enumerate}
\item Alice picks $X \in_R \01^n$ and $\Theta \in_R \{+,\times\}^n$.
\item For $i=1,\ldots,n$: In time slot $t=i$, Alice sends
  $\ket{X_i}_{\Theta_i}$ as a phase- or polarization-encoded weak pulse
  of light to Bob.
\item \label{step:reception} In each time slot, Bob measures the
  incoming qubit in the basis corresponding to his choice bit $C \in
  \{+,\times\}$ and records whether he detects a photon or not.  He
  obtains some bit-string $X' \in \01^m$ with $m \leq n$.
\item Bob reports back to Alice in which time slots he received a
  qubit. Alice restricts herself to the set of $m < n$ bits that Bob
  did not report as missing. Let this set of qubits be $S_{\rm
    remain}$ with $|S_{\rm remain}|=m$.
\item \label{step:alicecheck} Let $\setI_b = \{i \in S_{\rm
    remain}\mid \Theta_i = b\}$ for $b \in \{+,\times\}$ and let
  $m_b=|\setI_b|$.  Alice aborts the protocol if either $m_+$ or
  $m_\times$ are outside the interval $[(1-p_{\rm erase}-\eps)n/2,(1-p_{\rm erase}+\eps)n/2]$.
  If this is not the case, Alice picks two two-universal hash functions
  $F_+,F_\times \in_R \setF$. At time $t=n+T_{\rm rev}$, Alice sends
  $\setI_+$,$\setI_\times$, $F_+$,$F_\times$, and the syndromes $\syn(X_{|\setI_+})$
  and $\syn(X_{|\setI_\times})$ according to codes of appropriate length $m_b$ to Bob.  Alice outputs $S_+ =
  F_+(X_{|\setI_+})$ and $S_{\times} = F_\times(X_{|\setI_\times})$.
\item Bob uses $\syn(X_{|\setI_C})$ to correct the errors on his output $X'_{|\setI_C}$. He obtains the corrected bit-string
$X_{\rm cor}$ and outputs $S'_C = F_C(X_{\rm cor})$.
\end{enumerate}
\end{protocol}

\subsection{Security Analysis}

\paragraph{Correctness} By assumption, $p_{\rm erase}$ is independent
for every pulse and independent of the basis in which Alice sent the
qubits. Thus, by Chernoff's Inequality (Lemma~\ref{lem:chernoff}),
$S_{\rm remain}$ is, except with negligible probability, a random
subset of $m$ qubits independent of the value of $\Theta$ and such
that $(1-p_{\rm
  erase} - \eps)n \leq m \leq (1-p_{\rm erase} + \eps)n$ . This
implies that in Step~\ref{step:alicecheck} the protocol is aborted
with a probability only exponentially small in $n$. The codes are
chosen such that Bob can decode except with negligible probability.
These facts imply that if both parties are honest, the protocol is
correct (i.e. $S_C=S'_C$) with exponentially small probability of
error.

\paragraph{Security against Dishonest Alice}
Even though in this scenario Bob {\em does} communicate to Alice,
the information about which qubits were erased is (by assumption)
independent of the basis in which he measured and thus of his choice
bit $C$. Hence Alice does not learn anything about his choice bit
$C$. Her input strings can be extracted as in the analysis of
Protocol~\ref{prot:nonoise}.

\paragraph{Security against Dishonest Bob}

We prove the following:

\begin{theorem}\label{thm:tradeoffPractical}
Protocol~\ref{prot:practical} is secure against a dishonest receiver
Bob with error of at most $2\eps$, if $n \geq \frac{8}{5} \log
\frac{2}{\eps^4}$ and \beq \ell \leq  \left(t -
\delta -h(p_{\rm
    error})\right)(1-p_{\rm erase})\frac{n}{4} - \eps \frac{n}{2} +
\frac{1}{2} - \log\left(\frac{1}{\eps}\right) \, , \label{eq:lbound}
\eeq where $\delta = 8\sqrt{\log(2/\eps^4)/((1-p_{\rm
erase}-\eps)n)}$, and $t$ is the uncertainty bound on the
conditional Shannon entropy fulfilling Eq.~\eqref{eq:vNlower}.
\end{theorem}
\begin{sketch}
  First of all, we note that Bob can always make Alice abort the
  protocol by reporting back an insufficient number of received
  qubits. If Alice does not abort the protocol in
  Step~\ref{step:alicecheck}, we have that $(1-p_{\rm erase} -
  \eps)n/2 \leq m_+,m_\times \leq (1-p_{\rm erase} + \eps)n/2$. We
  define $D$ as in the security proof of Protocol~\ref{prot:nonoise}. The
  security analysis is the same, but we need to subtract the amount
  of error correcting information $|\syn(X_{|\setI_{\ol{D}}})|$ from
  the entropy of the dishonest receiver. If Alice does not abort the
  protocol in Step~\ref{step:alicecheck}, we have that
  $|\syn(X_{|\setI_{\ol{D}}})| \leq h(p_{\rm error}) (1-p_{\rm erase}
  + \eps)n/2$. Hence,
\begin{align*}
&\hmine{\eps}{X_{\ol{D}}|\Theta F_D D S_D \syn(X_{|\setI_{\ol{D}}}) K E} \\
&\geq \hmine{\eps}{X_{\ol{D}}D S_D \syn(X_{|\setI_{\ol{D}}}) |\Theta F_D K E} -
(\ell + 1) - h(p_{\rm error}) m/2
\\
&\geq (t - \delta)(1-p_{\rm erase}-\eps)n/2 - (\ell+1) - h(p_{\rm
  error})(1-p_{\rm erase}+\eps)n/2 - 1 - \ell\\
&\geq (t-\delta - h(p_{\rm error})(1-p_{\rm erase})n/2 -
\underbrace{(t-\delta+h(p_{\rm error}))}_{\leq 2}\eps n/2 - 1 - \ell
\, ,
\end{align*}
where $(t - \delta+h(p_{\rm error})) \leq 2$ since $t \leq 1$. Using
this inequality to bound the security parameter via the privacy
amplification Theorem \ref{thm:PA} gives the claimed bound on
$\ell$, Eq.~(\ref{eq:lbound}).
\end{sketch}

\paragraph{Remarks} Note that it is only possible to choose a code $C$
that satisfies the stated parameters {\em asymptotically}. For a
real---finite block-length---code, deviations from this asymptotic
behavior need to be taken into account. For the sake of clarity we
have omitted these details in the analysis above. Secondly, the
dishonest parties need to obtain an estimate for $p_{\rm error}$
prior to the protocol. One approach would be to use a worst case
estimate based what is possible with present-day technology.
Alternatively, one could follow Step~1 of the protocol
in~\cite{crepeau:practicalOT} as suggested above. However, one needs
to analyze this estimation procedure in a practical setting.
Thirdly, when weak photon sources are used in this protocol, one
needs to analyze the security threat due to the presence of
multi-photon emissions which Bob can exploit in
photon-number-splitting attacks as in QKD. See~\cite{arxiv:noisy}
for a first discussion of the effect of such attacks.

\subsection{Depolarizing Noise}
As an example, we again consider the security trade-off when Bob's
storage is affected by depolarizing noise. It follows directly from Theorems~\ref{thm:PA},
\ref{thm:tradeoffPractical} and \ref{thm:depolarize} that
\begin{corollary}\label{thm:depolTradeoff}
  Let $\mN$ be the depolarizing quantum operation given by
  Eq.~(\ref{eq:depolChannel}).  Then the protocol can be made secure
  (by choosing a sufficiently large $n$) as long as
\begin{eqnarray*}
h\left(\frac{1+r}{2}\right) > h(p_{\rm error}) && \mbox{ for } r \geq \hat{r}, \\
1/2  > h(p_{\rm error})
&& \mbox{ for } r < \hat{r},
\end{eqnarray*}
where $\hat{r} \assign 2 h^{-1}(1/2)-1 \approx 0.7798$.
\end{corollary}

Hence, our security parameters are greatly improved from our
previous analysis~\cite{arxiv:noisy}.  For $r < \hat{r}$ we can now
obtain security as long as the quantum bit error rate $p_{\rm error}
\lessapprox 0.11$, compared to 0.029 before. For the case of $r \geq
\hat{r}$, we can essentially show security as long as the noise on
the channel is strictly less than the noise in Bob's quantum
storage.  Note that we cannot hope to construct a protocol that is
both correct and secure when the noise of the channel exceeds the
noise in Bob's quantum storage.
However, it remains an open question whether it is possible to
construct a protocol or improve the analysis of the current protocol
such that security can be achieved even for very small $n$.

Corollary~\ref{thm:depolTradeoff} puts a restriction on the noise
rate of the honest protocol. Yet, since our protocols are
particularly interesting at short distances (e.g.~in the case of
secure identification we describe below), we can imagine free-space
implementations over very short distances such that depolarization
noise during transmission is negligible and the main noise source is
due to Bob's honest measurements.

In the near-future, if good photonic memories become available (see
e.g.~\cite{julsgaard+:mem, boozer+:qmem, chaneliere+:qmem,
  eisaman+:qmem, rosenfeld+:qmem, pittman_franson:memory} for recent
progress), we may anticipate that storing the qubit is a better
attack than a direct measurement.  Note, however, that we are free
in our protocol to stretch the reveal time $T_{\rm rev}$ between
Bob's reception of the qubits and his reception of the classical
basis information, say, to seconds, which means that one has to
consider the overall noise rate on a qubit that is stored for
seconds.

In terms of long-term security, {\em fault-tolerant} photonic
computation (e.g., with the KLM scheme \cite{KLM:lo}) might allow a
dishonest Bob to encode the incoming quantum information into a
fault-tolerant quantum memory. Such an encoding would guarantee that
the effective noise rate in storage can be made arbitrarily small. The
encoding of a single unknown state is \emph{not} a fault-tolerant
quantum operation however. Hence, even in the presence of a quantum
computer, there is a residual storage noise rate due to the
unprotected encoding operation. The question of security then becomes
a question of a trade-off between this residual noise rate versus the
intrinsic noise rate for honest parties.
Intuitively, it might be possible to arrange the setting such that
tasks of honest players are always technically easier (and/or cheaper)
to perform than the ones for dishonest players. Possibly, this
intrinsic gap can be exploited for cryptographic purposes. The current
paper can be appreciated as a first step in this direction.

\section{Extension to Secure Identification}\label{sec:others}
In this section, we like to point out how our model of noisy quantum
storage with individual-storage attacks also applies to protocols
that achieve more advanced tasks such as secure identification. The
protocol from~\cite{DFSS07} allows a user $U$ to identify
him/herself to a server $S$ by means of a personal identification
number (PIN). This task can be achieved by securely evaluating the
equality function on the player's inputs. In other words, both $U$
and $S$ input passwords $W_U$ and $W_S$ into the protocol and the
server learns as output whether $W_U = W_S$ or not. The protocol
proposed in~\cite{DFSS07} is secure against an unbounded user $U$
and a quantum-memory bounded server $S$ in the sense that it is
guaranteed that if a dishonest player starts with quantum side
information which is uncorrelated with the honest player's password
$W$, the only thing the dishonest player can do is guess a possible
$W'$ and learn whether $W=W'$ or not while not learning anything
more than this mere bit of information about the honest user's
password $W$. This protocol can also be (non-trivially) extended to
additionally withstand man-in-the-middle attacks.

The security proof against a quantum-memory bounded dishonest server
(and man-in-the-middle attacks) relies heavily on the uncertainty
relation first derived in~\cite{serge:new} and used for proving the
security of 1-2 OT. This uncertainty relation guarantees a lower
bound on the smooth min-entropy of the encoded string $X$ from the
dishonest player's point of view. As we establish a similar type of
lower bound (Cor.~\ref{cor:AEP} and Eq.~(\ref{eq:vNlower})) on the
smooth min-entropy in the noisy-storage model, the security proof
for the identification scheme (and its extension) translates to our
model.

In terms of the proof of Proposition~3.1 of~\cite{DFSS07},
the pair $X_i,X_j$ has essentially $t \cdot d$ bits of min-entropy
given $\Theta,K,$ and $E$, where $t$ is the uncertainty lower bound on the
conditional Shannon entropy from Eq.~(\ref{eq:vNlower}) and $d$ is
the minimal distance of the code used in the identification scheme.
Lemma~\ref{lemma:ESLident} implies that there exists $W'$ (called $V$ in
Lemma~\ref{lemma:ESLident}) such that if $W \neq W'$ then $X_W$ has
essentially $t d/2 - \log(m)$ bits of min-entropy given
$W,W',\Theta,K,E$. Privacy amplification then guarantees that
$F(X_W)$ is $\eps'$-close to uniform and independent of
$F,W,W',\Theta,K,E$, conditioned on $W \neq W'$, where $\eps' =
\frac12 2^{-\frac12(t d/2 - \log(m) - \ell)}$.
Security against a dishonest server with noisy quantum storage
follows as in~\cite{DFSS07} for an error parameter $\eps$
which is exponentially small in $t d - 2\log(m) - 2 \ell$.

\section{Conclusion}
We have obtained improved security parameters for oblivious transfer
in the noisy-quantum-storage model. Yet, it remains to prove security
against general coherent noisy attacks.  The problem with analyzing a
coherent attack of Bob described by some quantum operation ${\cal S}$
affecting all his incoming qubits is not merely a technical one: one
first needs to determine a realistic noise model in this
setting.
Symmetrizing the protocol as in the proof of QKD \cite{renato:diss}
and using de Finetti type arguments does not immediately work here.
However, one can analyze a specific type of coherent noise, one that
essentially corresponds to an eavesdropping attack in QKD. Note that
the 1-2 OT protocol can be seen as two runs of QKD interleaved with
each other. The strings $f(x_{|\setI_+})$ and $f(x_{|\setI_\times})$
are then the two keys generated. The noise must be such that it
leaves Bob with exactly the same information as the eavesdropper Eve
in QKD. In this case, it follows from the security of QKD that the
dishonest Bob (learning exactly the same information as the
eavesdropper Eve) does not learn anything about the two keys.

Clearly, there is a strong relation between QKD and the protocol for
1-2 OT, and one may wonder whether other QKD protocols can be used
to perform oblivious transfer in our model. Intuitively, this is indeed
the case, but it remains to evaluate explicit parameters for the security
of the resulting protocols.

It will be interesting to extend our results to a security analysis of
a noise-robust protocol in a realistic physical setting, where, for
example, the use of weak laser pulses allows the possibility of
photon-number-splitting attacks. Such a comprehensive security
analysis has been carried out in \cite{GLLP:qkd} for quantum key
distribution.

\section*{Acknowledgments}
We thank Robert K\"onig and Renato Renner for useful discussions about
the additivity of the smooth min-entropy and the permission to include
Lemma~\ref{app:additivity}. CS is supported by EU fifth framework
project QAP IST 015848 and the NWO VICI project 2004-2009. SW is
supported by NSF grant number PHY-04056720.

\bibliographystyle{alpha}
\bibliography{ot_aa_long}

\appendix

\section{Appendix: Properties of The Conditional Smooth Min-Entropy}
\label{app:a}

In this Appendix we provide the technical proofs of the Lemmas and the Theorem in Section \ref{sec:propsmooth}.
We restate the claims for convenience.

\subsection{Proof of Lemma~\ref{lem:chain} (Chain Rule)} \label{app:chain}
\begin{lemma}[Chain Rule]
For any ccq-state $\rho_{XYE} \in \states(\hil_{XYE})$
and for all $\eps \geq 0$, it holds that
$$
\hmine{\eps}{X | Y E } \geq \hminee{XY | E} - \log|\setY|,
$$
where $|\setY|$ is the alphabet size of the random variable $Y$.
\end{lemma}
\begin{proof}
 For $\eps=0$, it follows from Eq.~\eqref{eq:duality} that we need to show that
\begin{equation} \label{eq:chain}
 \pg(XY|E) \geq \pg(X|YE) \cdot \frac{1}{|\setY|} \, .
\end{equation}
For a given value $y$, let $\set{M_x^y}_x$ be the POVM on register
$E$ which optimally guesses $X$ given $Y$. A particular strategy of
guessing \emph{$X$ and $Y$} from $E$ is to guess a value of $y$
uniformly at random from $\setY$ and subsequently measure $E$ with
the POVM $\set{M_x^y}_x$. The success probability of this strategy
is exactly the r.h.s of~\eqref{eq:chain}. Clearly, the optimal
guessing probability $\pg(XY|E)$ can only be better than this
particular strategy. For $\eps >0$, let $\hat{\rho}_{XYE} \in
\ball(\rho_{XYE})$ be the state in the $\eps$-ball around
$\rho_{XYE}$ that maximizes the min-entropy $\hminee{XY|E}$. The
technique from Remark~3.2.4 in~\cite{renato:diss} can be used to
show that $\hat{\rho}_{XYE}$ is a ccq-state. By the derivation above
for $\eps=0$, we obtain that
\begin{align*}
\pg(XY|E)_{\hat{\rho}} &\geq \pg(X|YE)_{\hat{\rho}} \cdot \frac{1}{|\setY|}  \\
&\geq \min_{\tilde{\rho}_{XYE} \in \ball(\rho_{XYE})} \pg(X|YE)_{\tilde{\rho}} \cdot \frac{1}{|\setY|} \, ,
\end{align*}
which proves the lemma by taking the negative logarithms and using
Eq.~\eqref{eq:duality}.
\end{proof}

\subsection{Proof of Lemma~\ref{lem:additivity}
  (Additivity)} \label{app:additivity}

To show additivity of the smooth min-entropy we will employ semidefinite programming, where we refer to~\cite{boyd:convex} for 
in-depth information. Here, we will use semidefinite programming in the language of~\cite{KRS09} to express
the primal and dual optimization problem given by parameters $c \in \mV_1$ and $b \in \mV_2$ in vector spaces $\mV_1$ and $\mV_2$
with inner products $\spr{\cdot}{\cdot}_1$ and $\spr{\cdot}{\cdot}_2$. We will optimize over variables $v_1 \in K_1$ and $v_2 \in K_2$, 
where $K_1 \subset \mV_1$ and $K_2 \subset \mV_2$ are convex cones in the respective vector spaces. In our application below, these will
simply be the cones of positive-semidefinite matrices. 
We can then write
\begin{align}
\gamma^{\textrm{primal}}=\min_{\substack{
v_1\geq 0 \\
Av_1\geq  b
}} \spr{v_1}{c}_1\qquad\textrm{ and }\qquad
\gamma^{\textrm{dual}}=\max_{\substack{
v_2\geq 0 \\
A^*v_2 \leq c
}} \spr{b}{v_2}_2\label{eq:primalanddual},
\end{align}
where $A: \mV_1 \rightarrow \mV_2$ is a linear map defining the particular problem we wish to solve.
We use $A^*:\mV_2 \rightarrow \mV_1$ to denote its dual map satisfying
\begin{align*}
\spr{A v_1}{v_2}_2&=\spr{v_1}{A^*v_2}_1\qquad\textrm{ for all }v_1\in\mV_1, v_2\in\mV_2\ .
\end{align*}
Note that we have $\gamma^{\textrm{primal}} \geq
\gamma^{\textrm{dual}}$ by weak duality. In this case our SDPs will be strongly feasible, 
giving us $\gamma^{\textrm{primal}} = \gamma^{\textrm{dual}}$ known as strong duality.
Our proof is based on the same idea as~\cite[Lemma 2]{arxiv:noisy} applied
to the smoothed setting.
We thank Robert K{\"o}nig for allowing us to include the following.

\begin{lemma}[Additivity (K{\"o}nig and Wehner)]
Let $\rho_{AB}$ and $\rho_{A'B'}$ be two independent qq-states. For $\eps
\geq 0$, it holds that
\[ \hmine{\eps^2}{AA'|BB'} \leq \hmine{\eps}{A|B}+\hmine{\eps}{A'|B'}\, .
\]
\end{lemma}
\begin{proof}
In order to prove additivity, it is important to realize that the smooth conditional min-entropy can be
written as semi-definite program:

\begin{align} \hminee{A|B}
&= \max_{\hat{\rho}_{AB} \in \ball(\rho_{AB})}
\hmin(A|B)_{\hat{\rho}} \nonumber\\
&=  \max_{{\hat{\rho}_{AB} \in \ball(\rho_{AB})}}
-\log \;
\min_{\substack{
\sigma_B \geq 0\\
\hat{\rho_{ABC}} \geq 0\\
\id_A \otimes \sigma_B \geq
\hat{\rho}_{AB}
}} \Tr(\sigma_B) \\
&=  -\log \;
\min_{\substack{
\sigma_B \geq 0\\
\hat{\rho_{ABC}} \geq 0\\
\hat{\rho}_{AB} \in \ball(\rho_{AB}) \\ 
\id_A \otimes \sigma_B
\geq \hat{\rho}_{AB}
}} \Tr(\sigma_B) \, .
\end{align}
where $\sigma_B \in \pos(\hil_B)$ throughout.
Let $\ket{\psi}_{ABC}$ be a purification of $\rho_{AB}$. Then, all
states $\hat{\rho}_{AB} \in \ball(\rho_{AB})$ can be obtained by an
extension $\hat{\rho}_{ABC} \geq 0$ such that $\Tr(\hat{\rho}_{ABC})
\leq 1$, and $\Tr(\hat{\rho}_{ABC} \proj{\psi_{ABC}}) \geq 1-\oeps$
with $\oeps = \eps^2$.
Therefore, we can write
\begin{align*} \hminee{A|B}
&=  -\log \;
\min_{\substack{
\Tr(\hat{\rho}_{ABC} \proj{\psi_{ABC}}) \geq 1-\oeps \\ 
1 \geq \Tr(\hat{\rho}_{ABC}) \\ \id_A \otimes \sigma_B \geq \hat{\rho}_{AB} 
}} \Tr(\sigma_B) 
\ ,
\end{align*}
where the minimum is taken over all $\sigma_B \in \pos(\hil_B)$ and
$\hat{\rho}_{ABC} \in \pos(\hil_{ABC})$, which is a semi-definite
program (SDP). Our goal will be to determine the dual of this
semidefinite program which will then allow us to put an upper bound on
the smooth min-entropy as desired.

\newcommand{\herm}{\operatorname{Herm}}
We now first show how to convert the primal of this semidefinite
program into the form of Eq.~(\ref{eq:primalanddual}).  Let $\mV_1 =
\herm(\hil_B) \oplus \herm(\hil_{ABC})$ where $\herm(\hil)$ is the
  (real) vector space of Hermitian operators on $\hil$. Let $K_1
  \subset \mV_1$ be the cone of positive semi-definite operators.
% $\mV_1 =
% \bop(\hil_B) \oplus \bop(\hil_{ABC})$ 
% % \Real \oplus \spann\{\sigma_B \oplus \hat{\rho}_{ABC}\}
% , $K_1$ be the non-negative operators on $\mV_1$,
  Let $c = \id_B \oplus 0_{ABC}$ where $0_{ABC}$ is the zero-operator
  on $\hil_{ABC}$. Let the inner product be defined as
  $\spr{v_1}{v_1'}_{1}=\Tr(v_1^\dagger v_1')$. Note that this allows
  us to express our objective function as
$$
\spr{\sigma_B\oplus\hat{\rho}_{ABC}}{c}_1=\Tr(\sigma_B).
$$
It remains to rewrite the constraints in the appropriate form. To this
end, we need to define $\mV_2 = \Real \oplus \Real \oplus
\herm(\hil_A) \oplus \herm(\hil_B)$, $K_2 \subset \mV_2$ the cone of
positive semi-definite operators and take the inner product to have
the same form $\spr{v_2}{v_2'}_{2}=\Tr(v_2^\dagger v_2')$. We then let
$b \in \mV_2$ be given as
$$
b = (1-\oeps) \oplus (-1) \oplus 0_{AB}\ ,
$$
and define the map
$$
A(\sigma_B \oplus \hat{\rho}_{ABC}) = 
\Tr(\hat{\rho}_{ABC}\proj{\psi_{ABC}}) \oplus
(-\Tr(\hat{\rho}_{ABC})) \oplus (\id_A\otimes\sigma_B-\hat{\rho}_{AB})\ .
$$
Note that $v_1 = \sigma_B \oplus \hat{\rho}_{ABC} \geq 0$ and $A(v_1) \geq b$ now exactly represent our constraints.

We now use this formalism to find the dual. Note that we may write any
$v_2 \in \mV_2$ with $v_2 \geq 0$ as $v_2 = r \oplus s \oplus Q_{AB}$ where
$Q_{AB} \in \pos(\hil_A \otimes \hil_B)$ and $r,s \in \Real$. To find
the dual map $A^*$ note that
\begin{align*}
\spr{Av_1}{v_2}_2 &=r\Tr(\hat{\rho}_{ABC}\proj{\psi_{ABC}})-s
\Tr(\hat{\rho}_{ABC}) 
+\Tr(Q_{AB}(\id_A\otimes\sigma_B-\hat{\rho}_{AB}))\\
  &=r\Tr(\hat{\rho}_{ABC}\proj{\psi_{ABC}})- s
\Tr(\hat{\rho}_{ABC}) +\Tr(Q_{B}\sigma_B)-\Tr((Q_{AB}\otimes
  \id_C)\hat{\rho}_{ABC}),
\end{align*}
and we therefore have
\begin{align*}
A^*(v_2) &=(0_B\oplus r \proj{\psi_{ABC}})-(s \id_{ABC})+(Q_B\oplus 0_{ABC})-(0_B\oplus Q_{AB}\otimes\id_C)  \ ,
\end{align*}
which is all we require using Eq.~(\ref{eq:primalanddual}). To find a more intuitive interpretation of the dual note
that $A^*(v_2) \leq c$ is equivalent to
\begin{align}
\id_B &\geq Q_B\ ,\\
Q_{AB}\otimes \id_C&\geq r\proj{\psi_{ABC}} - s \id_{ABC}\ ,\label{eq:qbm}
\end{align}
and $\spr{b}{v_2}_2 =r(1-\oeps)-s$. The dual can thus be written as
\begin{align*}
\gamma^{\textrm{dual}}&=\max_{\substack{r\geq 0, s \geq 0\\ \id_B \geq Q_B\\
    Q_{AB}\otimes \id_C\geq r\proj{\psi_{ABC}}-s \id_{ABC}} }r(1-\oeps) - s\ .
\end{align*}

We now use the dual formulation to upper bound the smooth min-entropy
of the combined state $\rho_{AB} \otimes \rho_{A'B'}$ and parameter
$\tilde{\oeps}$ by finding a lower bound to the dual semidefinite
program.  Let $\tilde{\gamma}(\tilde{\oeps})$ denote the optimal
solution of the dual of the SDP for the combined state for error $\tilde{\oeps}$.  For each individual
state, we may solve the above SDP, where we let $Q_{AB}, r$ and $s$
denote the optimal solution for state $\rho_{AB}$ with parameter
$\oeps$ and optimal value $\gamma(\oeps)$, and let $Q_{A'B'},r'$ and
$s'$ denote the optimal solution for state $\rho_{A'B'}$ with
parameter $\oeps'$ and optimal value $\gamma(\oeps')$.  We now use
these solutions to construct a solution (not necessarily the optimal
one) for the combined state $\rho_{AB} \otimes \rho_{A'B'}$.  Let
$\tilde{Q}=Q_{AB}\otimes Q_{A'B'}$, $\tilde{r}=rr'$ and
$\tilde{s}=rs'(1-\delta)+sr'(1-\delta')-ss'$. Note that $rs' \geq 0$
and $r'(1-\delta')-s' \geq 0$ for the optimal $r',s'$ and hence
\begin{align*}
\tilde{r}&\geq 0\ ,\ \tilde{s}\geq 0\ ,\\
\id_{BB'} &\geq
\tilde{Q}_{BB'}\ ,\\
\tilde{Q}_{AA'BB'}\otimes\id_{CC'}&\geq (r\, \proj{\psi_{ABC}} -
s\, \id_{ABC})\otimes (r'\,\proj{\psi_{A'B'C'}} - s\, \id_{A'B'C'}) \\
&\geq \tilde{r}\, \proj{\psi_{ABC}} \otimes \proj{\psi_{A'B'C'}} -
\tilde{s}\, \id_{ABC} \otimes \id_{A'B'C'} \ ,
\end{align*}
and thus $\tilde{Q}$ is indeed a feasible solution for the combined problem. 
Choosing $\tilde{\oeps}$ as
\begin{align*}
\tilde{\oeps} &=\oeps+\oeps'-\oeps\oeps'
\end{align*} we have
\begin{align*}
\tilde{\gamma}(\tilde{\oeps})&\geq \tilde{r}(1-\tilde{\oeps}) - \tilde{s} =
\gamma(\delta) \gamma'(\delta')\ .
\end{align*}
We hence obtain
\begin{align*}
  \hmin^{\sqrt{\tilde{\oeps}}}(\tilde{A}|\tilde{B})&\leq
  \hmin^{\sqrt{\oeps}}(A|B)+ \hmin^{\sqrt{\oeps'}}(A'|B') \ . \end{align*} 
For $\oeps=\oeps'$, we have
\begin{align*}
\tilde{\oeps}&=2\oeps-\oeps^2\geq \oeps^2\ .
\end{align*}
Putting everything together we thus have
\begin{align*}
\hmin^{\oeps}(\tilde{A}|\tilde{B})&\leq \hmin^{\sqrt{\oeps}}(A|B)+ \hmin^{\sqrt{\oeps}}(A'|B')\ ,
\end{align*}
from which the result follows since $\oeps = \eps^2$.
\end{proof}

\subsection{Proof of Lemma~\ref{lem:monotone}   (Monotonicity)}  \label{app:monotonicity}
\begin{lemma}[Monotonicity]
For a ccq-state $\rho_{XYE}$ and for all $\eps \geq 0$, it holds that
\begin{equation} \nonumber
\hmine{\eps}{XY | E}  \geq \hminee{Y | E} \, .
\end{equation}
\end{lemma}

\begin{proof}
For $\eps=0$, the lemma follows from Eq.~(\ref{eq:duality}), that
is, guessing $XY$ from $E$ is harder than guessing only $Y$ from $E$
and therefore, $\pg(XY|E) \leq \pg(Y|E)$.

For $\eps > 0$ the idea behind the argument is similar. Let the
maximum in $\hminee{Y|E}$ be achieved by a density matrix
$\hat{\rho}_{YE}$, i.e. $\hminee{Y|E}= \hmin(Y|E)_{\hat{\rho}}$ such that $C(\rho_{YE},\hat{\rho}_{YE}) \leq \eps$ and $\Tr(\hat{\rho}_{YE}) \leq 1$. Remark~3.2.4 in~\cite{renato:diss} shows that
$\hat{\rho}_{YE}$ is a cq-state. We can express this min-entropy in
terms of the guessing probability, Eq.~(\ref{eq:duality}), and thus
\begin{equation}
\hminee{Y|E}_{\hat{\rho}}=-\log \pg(Y|E)_{\hat{\rho}_{YE}} \leq -\log
\pg(XY|E)_{\hat{\rho}_{XYE}} \, ,
\label{eq:bounds}
\end{equation}
where $\hat{\rho}_{XYE}$ is any ccq-state which has $\hat{\rho}_{YE}$ as
its reduced state, i.e $\Tr_X(\hat{\rho}_{XYE})=\hat{\rho}_{YE}$.  Now
we would like to show that one can choose an extension
$\hat{\rho}_{XYE}$ such that $C(\rho_{XYE},\hat{\rho}_{XYE}) = \sqrt{1-F(\rho_{XYE},\hat{\rho}_{XYE})^2} \leq \eps$ and $\Tr (\hat{\rho}_{XYE}) \leq 1$.
If we can determine such an extension,
we can upper-bound the r.h.s. in Eq.~(\ref{eq:bounds}) by $\hminee{XY|E}$ which is the supremum of $-\log \pg(XY|E)$ over
states in the $\eps$-neighborhood of $\rho_{XYE}$. This would
prove the Lemma.

Let $\ket{\Psi}_{XYEC}$ be a purification of $\rho_{XYE}$ and hence also a purification of $\rho_{YE}$. By
Uhlmann's theorem (see e.g. \cite{nielsen&chuang:qc}), we have for
the fidelity $F(\rho_{YE},\hat{\rho}_{YE})$ between $\rho_{YE}$ and
$\hat{\rho}_{YE}$ that
$$
F(\rho_{YE}, \hat{\rho}_{YE})=\max_{{\ket{\Psi'}}_{XYEC}} | \braket{\Psi}{\Psi'} | \assign F(\proj{\Psi}, \proj{\hat{\Psi}}) \, ,
$$
where $\ket{\hat{\Psi}}_{XYEC}$ is the purification of $\hat{\rho}_{YE}$ achieving the maximum.
The monotonicity property of the fidelity under taking the partial trace gives
$$
F(\proj{\Psi} ,\proj{\hat{\Psi}} ) \leq
F(\Tr_C(\proj{\Psi}), \Tr_C(\proj{\hat{\Psi}}))=F(\rho_{XYE},\hat{\rho}_{XYE}) \, ,
$$
where $\hat{\rho}_{XYE} \assign \Tr_C(\proj{\hat{\Psi}_{XYEC}})$. Hence
\begin{equation}\label{eq:fidLower}
\sqrt{1-\eps^2} \leq F(\rho_{YE}, \hat{\rho}_{YE}) \leq F(\rho_{XYE},\hat{\rho}_{XYE}) \, ,
\end{equation}
and therefore, $C(\rho_{XYE},\hat{\rho}_{XYE}) \leq \eps$. If $\Tr(\hat{\rho}_{XYE}) > 1$, it follows that also $\Tr(\hat{\rho}_{YE}) >1$ which contradicts the assumption. Therefore, it must be the case that $\Tr(\hat{\rho}_{XYE}) \leq 1$.

It remains to show that $\hat{\rho}_{XYE}$ is a ccq-state. Because of
\begin{align*}
F(\rho_{YE}, \hat{\rho}_{YE}) &= F(\proj{\Psi} ,\proj{\hat{\Psi}} )\\
& \leq
F(\Tr_C(\proj{\Psi}), \Tr_C(\proj{\hat{\Psi}}))=F(\rho_{XYE},\hat{\rho}_{XYE}) \leq F(\rho_{YE},\hat{\rho}_{YE}) \, ,
\end{align*}
these quantities are all equal and in particular, we could do a measurement on the $X$-register of $\hat{\rho}_{XYE}$ without increasing the fidelity. Hence, we can assume the optimal purification $\proj{\hat{\Psi}_{XYEC}}$ is such that $\hat{\rho}_{XYE}$ is a ccq-state.
\end{proof}

\subsection{Proof of Theorem~\ref{thm:aequi}}\label{app:aequi}
\begin{theorem}
For $i=1,\ldots,n$, let $\rho_{i} \in \mS(\hil_{AB})$ be density operators.
Then, for any $\eps > 0$,
\begin{equation*}
\hminee{A^n|B^n}_{\bigotimes_{i=1}^n \rho_{i}} \geq \sum_{i=1}^n \left[ \H(A_i|B_i)_{\rho_i} \right] -
  \delta(\eps,\gamma)\sqrt{n}
\, ,
\end{equation*}
where, for $n \geq \frac{8}{5} \log \frac{2}{\eps^2}$, the error is given by
\[ \delta(\eps,\gamma) \assign 4 \log \gamma \sqrt{\log \frac{2}{\eps^2}}
\]
and the single-system entropy contribution by
\[ \gamma \leq 2 \max_i \sqrt{\rank(\rho_{A_i})} + 1 \, .
\]
\end{theorem}

\begin{proof}
The proof is analogous to the proof of Theorem~7 in~\cite{TCR08}.
For convenience, we point out where their proof needs to be adapted.
We need the following definitions. Let $\hil'_{AB}$ be a copy of
$\hil_{AB}$ and let $\ket{\gamma} \assign \sum_i \ket{i} \otimes
\ket{i}$ be the unnormalized fully entangled state on $\hil_{AB}
\otimes \hil'_{AB}$. Define the purification $\ket{\phi} \assign
(\sqrt{\rho_{AB}} \otimes \id_{AB}) \ket{\gamma}$ of $\rho_{AB}$ and
let $1< \alpha \leq 2$, $\beta \assign \alpha -1$, and $X \assign
\rho_{AB} \otimes (\id_A \otimes \rho_B^{-1})^T$. The conditional
$\alpha$-entropy is defined as $\H_\alpha(A|B)_{\rho|\sigma} \assign
\frac{1}{1-\alpha} \log \Tr (\rho_{AB}^\alpha (\id_A \otimes
\sigma_B)^{1-\alpha})$. The authors of \cite{TCR08} prove the
following lower bound
\begin{equation} \label{eq:AEP1}
\H_\alpha(A|B)_{\rho|\rho} \geq H(A|B)_\rho - \frac{1}{\beta \ln 2} \bra{\phi} r_\beta(X) \ket{\phi} \, ,
\end{equation}
where $r_\beta(t) \assign t^\beta - \beta \ln t -1$.

Let $\hat{\rho} = \rho_{AB}^1 \otimes \ldots \otimes \rho_{AB}^n$. Then, as in Equation~(27) of~\cite{TCR08}, we have
\begin{align}
\hminee{A^n|B^n}_{\hat{\rho}} \geq \hminee{A^n|B^n}_{\hat{\rho}|\hat{\rho}} \nonumber
&\geq \H_{\alpha}(A^n|B^n)_{\hat{\rho}|\hat{\rho}} - \frac{1}{\beta} \log \frac{2}{\eps^2}\\
&= \sum_{i=1}^n \H_\alpha(A|B)_{\rho^i|\rho^i} - \frac{1}{\beta} \log \frac{2}{\eps^2} \nonumber \\
&\geq \sum_{i=1}^n \left( \H(A|B)_{\rho^i} - \frac{1}{\beta \ln 2} \bra{\phi} r_\beta(X^i) \ket{\phi} \right)  - \frac{1}{\beta} \log \frac{2}{\eps^2} \, , \label{eq:AEP2}
\end{align}
where we used~\eqref{eq:AEP1} in the last step.

Let us define the single-system entropy contributions $\gamma^i \assign \bra{\phi} \sqrt{X^i} + 1/\sqrt{X^i} + \id \ket{\phi}$ of which we know that they are all $\geq 3$ and let $\gammamax$ be the largest of them. By choosing an appropriate $\mu \geq 0$ such that
$$ \beta = \frac{1}{2 \mu \sqrt{n}} \leq \sqrt{\frac{5}{8}} \; \frac{1}{2 \log \gammamax} \leq \min \left\{ \frac{1}{4}, \frac{1}{2 \log \gammamax} \right\} \, ,
$$
we can bound
$$
\frac{1}{\beta \ln 2} \bra{\phi} r_\beta(X^i) \ket{\phi} \leq \frac{2}{ \mu \sqrt{n} } \log^2(\gamma^i) \leq \frac{2}{ \mu \sqrt{n} } \log^2(\gammamax)
$$

Therefore, we can further lower bound~\eqref{eq:AEP2} as
\begin{align*}
\hminee{A^n|B^n}_{\hat{\rho}} &\geq \sum_{i=1}^n H(A|B)_{\rho^i} - \sum_{i=1}^n \frac{2}{\mu \sqrt{n}} \log^2(\gammamax) - 2 \mu \sqrt{n} \log \frac{2}{\eps^2} \\
&\geq \sum_{i=1}^n H(A|B)_{\rho^i} - 2 \sqrt{n} \left( \frac{1}{\mu} \log^2(\gammamax) + \mu \log \frac{2}{\eps^2} \right) \, .
\end{align*}
and the rest of the derivation goes as after Equation~(28) in~\cite{TCR08}.

In order to obtain the upper bound on $\gamma$, we notice that $\H_{1/2}(A|B)_{\rho|\rho} \leq H_{1/2}(A)_{\rho} \leq H_0(A)_{\rho} = \log (\rank(\rho_A))$.
\end{proof}

\section{Appendix: Proof of Theorem~\ref{thm:depolarize}} \label{app:mainproof}

\subsection{Setting the Stage}
We use the symmetries inherent in our problem to prove Theorem~\ref{thm:depolarize} in a
series of steps.

\begin{theorem}
Let $\mN$ be the depolarizing
quantum operation given by Eq.~\eqref{eq:depolChannel} and let $\H(X
| \Theta
  K E)$ be the conditional von Neumann entropy of one qubit. Then
$$
\H(X | \Theta K E) \geq \left\{
\begin{array}{ll}
h(\frac{1+r}{2}) & \mbox{ for } r \geq \hat{r} \, ,  \\
1/2  & \mbox{ for } r < \hat{r} \, ,
\end{array}\right.
$$
where $\hat{r} \assign 2 h^{-1}(1/2)-1 \approx 0.7798$.
\end{theorem}
In order to prove the theorem, we find Bob's strategy which
minimizes $\H(X|\Theta K E)$ as a function of the depolarizing noise
parameter $r$. As depicted in Figure~\ref{figure:depolModel}, in
each round the dishonest receiver Bob receives one of the four
possible BB84 states $\rho_{x\theta}$ at random. On such state he
may then perform any (partial) measurement $M$ given by measurement
operators $M=\{F_k\}$ such that $\sum_k F_k^\dagger F_k = \id$. For
clarity of notation, we do not use a subscript to indicate the round
$i$ as in the Figure. We denote by $\regE$ the register containing
the renormalized post-measurement state
$$
\rho_{x\theta}^{k,M}= \frac{F_k \rho_{x,\theta} F_k^\dag}{p_{k | x
\theta}},
$$
to which the depolarizing quantum operation ${\cal N}$ is applied. Here
$$
p_{k | x \theta}^M = \Tr( F_k \rho_{x \theta} F_k^\dag )
$$
is the probability to measure outcome $k$ when given state $\rho_{x
\theta}$. We omit the superscript $M$ if it is clear which
measurement is used. Note that we may write
$$
p_{x\theta k}^M = \frac{1}{4}p_{k|x \theta}^M \, ,
$$
$$
\sum_x p_{x\theta k}^M=p_{\theta k}^M  = \frac{1}{4}\Tr\left(F_k
\left(\rho_{0\theta} +
    \rho_{1\theta}\right)F_k^\dagger\right) =
\frac{1}{4}\Tr(F_kF_k^\dagger) \, ,
$$
and
$$
p_{x|\theta k}^M = \frac{p_{k|x\theta}^M}{4 p_{\theta k}^M}.
$$
Here we have used the fact that Alice chooses the basis and bit in
each round uniformly and independently at random.

First of all, note that for a cq-state $\rho_{YE} = \sum_y P_Y(y)
\ket{y}\bra{y} \otimes \rho_y^E$, the von Neumann entropy can be
expanded as
$$
\H(YE) = \H(Y) + \sum_y P_Y(y) \H(\rho_y^E) \, .
$$

Using this expansion, we can write
\begin{align} \nonumber
\H(X| \Theta K E)_M &= \H(X \Theta K E)_M - \H(\Theta K E)_M\\
&= \H(X \Theta
K)_M + \sum_{x \theta k} p_{x \theta k}^M \H\left(\mN\left(\rho_{x \theta}^{k,M}\right)\right) -
\H(\Theta K)_M - \sum_{\theta k} p_{\theta k}^M \H \left(\mN\left( \sum_x p_{x|\theta k}^M
\rho_{x \theta}^{k,M} \right)\right) \nonumber \\
&= \H(X | \Theta K)_M + \sum_{x \theta k} p_{x \theta k}^M \H\left(\mN\left(\rho_{x
  \theta}^{k,M}\right)\right) - \sum_{\theta k} p_{\theta k}^M \H \left(\sum_x p_{x|\theta k}^M
\mN\left(\rho_{x \theta}^{k,M}\right)\right). \label{eq:merit}
\end{align}
We use the notation $\H(X|\Theta K E)_M$ to emphasize that we
consider the conditional von Neumann entropy when Bob performed a
partial measurement $M$. In the following, we use the shorthand
$$
\B(M) \assign \H(X|\Theta K E)_M \, .
$$

\subsection{Using Symmetries to Reduce Degrees of Freedom}
Our goal is to minimize $\B(M)$ over all possible measurements $M =
\set{F_k}$ as a function of $r$. We proceed in three steps. First,
we simplify our problem considerably until we are left with a single
Hermitian measurement operator over which we need to minimize the
entropy. Second, we show that the optimal measurement operator is
diagonal in the computational basis. And finally, we show that
depending on the amount of noise, this measurement operator is
either proportional to the identity, or proportional to a rank one
projector.

First, we prove a property of the function $B(M)$ for a composition
of two measurements. Intuitively, the following statement uses the
fact that if we choose one measurement with probability $\alpha$ and
another measurement with probability $\beta$ our average success
probability is the average of the success probabilities obtained via
the individual measurements:
\begin{claim}\label{convexity}
Let $F=\{F_k\}_{k=1}^f$ and $G=\{G_k\}_{k=f+1}^{f+g}$ be two
measurements. Then, for $0 \leq \alpha \leq 1$ and a combined
measurement  $M = \alpha F + (1-\alpha) G \assign \set{\sqrt{\alpha}
    F_k}_{k=1}^{f} \cup \set{\sqrt{1-\alpha} G_k}_{k=f+1}^{f+g}$, we have
$$\B(\alpha F + (1-\alpha) G) = \alpha \B(F) + (1-\alpha) \B(G) \, .$$
\end{claim}
\begin{proof}
  Let $F=\set{F_k}_{k=1}^f$ and $G=\set{G_k}_{k=1}^g$ be measurements,
  $0 \leq \alpha \leq 1$ and let $M \assign \set{\sqrt{\alpha}
    F_k}_{k=1}^{f} \cup \set{\sqrt{1-\alpha} G_k}_{k=f+1}^{f+g}$.

It is easy to verify that we have the following relations for $1\leq k
\leq f$: $p_{x \theta k}^M= \alpha p_{x \theta k}^F$, $p_{x|\theta
k}^M=\frac{\alpha p_{x \theta k}^F}{\alpha p_{\theta k}^F} =
p_{x|\theta k}^F$, $p_{k | x \theta}^M = \alpha p_{k | x \theta}^F$
and $\rho_{x \theta}^{k,M} = \frac{\alpha F_k \rho_{x \theta}
  F_k^\dag}{p_{k|x \theta}^M} = \frac{F_k \rho_{x \theta}
  F_k^\dagger}{p_{k|x \theta}^F} = \rho_{x \theta}^{k,F}$
and analogously for $f+1 \leq k \leq f+g$.

We consider the three summands in Eq.~\eqref{eq:merit} separately.
For the first term we get
\begin{align*}
\H(X | \Theta K)_M &= \sum_{\theta k} p_{\theta k}^M
h\left(p_{0|\theta k}^M\right) \\
&= \sum_{\theta} \sum_{k=1}^f \alpha p_{\theta k}^F h\left(p_{0|\theta
  k}^F\right) + \sum_{\theta} \sum_{k=f+1}^{f+g} (1-\alpha) p_{\theta k}^G h\left(p_{0|\theta
  k}^G\right) \\
&=\alpha \H(X|\Theta K)_F + (1-\alpha) \H(X|\Theta K)_G \, .
\end{align*}
For the second term, we obtain
\begin{align*}
\sum_{x \theta k} p_{x \theta k}^M \H\left( \mN\left(\rho_{x
\theta}^{k,M}\right) \right) = \alpha \sum_{x \theta} \sum_{k=1}^f p_{x \theta k}^F
\H\left( \mN\left(\rho_{x \theta}^{k,F}\right) \right) + (1-\alpha) \sum_{x \theta} \sum_{k=f+1}^{f+g} p_{x
\theta k}^G \H\left(\mN\left(\rho_{x \theta}^{k,G}\right) \right) \, .
\end{align*}
The third term yields
\begin{align*}
\sum_{\theta k} p_{\theta k}^M &\H\left( \sum_x p_{x|\theta k}^M \mN\left(\rho_{x
\theta}^{k,M}\right) \right)\\
&= \alpha \sum_{\theta} \sum_{k=1}^f p_{\theta k}^F
\H\left(  \sum_x p_{x|\theta k}^F \mN\left(\rho_{x
\theta}^{k,F}\right) \right) + (1-\alpha) \sum_{\theta} \sum_{k=f+1}^{f+g} p_{x
\theta k}^G \H\left(  \sum_x p_{x|\theta k}^G \mN\left(\rho_{x
\theta}^{k,G}\right) \right) \, .
\end{align*}
\end{proof}

We can now make a series of observations.
\begin{claim}\label{invarianceOfBunderG}
Let $M = \{F_k\}$ and $G = \{\id,X,Z,XZ\}$. Then for all $g \in G$ we have $\B(M) = \B(gMg^\dagger)$.
\end{claim}
\begin{proof}
First of all, note
that for all $g \in G$, $g$ can at most exchange the roles of
$0$ and $1$. That is, we can perform a bit flip before the
measurement which we can correct for afterwards by applying
classical post-processing. Furthermore, since $g \in G$
is Hermitian and unitary we have
$$
p_{\theta k}^{gMg^\dag}=\frac{1}{4}
\Tr(g F_k g^\dag (\rho_{0 \theta} + \rho_{1 \theta}) g^\dag F_k^\dag g) =
\frac{1}{4} \Tr( F_k F_k^\dag ) = p_{\theta k}^M \, ,
$$
and hence there exists a bijection $f: \01 \rightarrow \01$ such that
$$
p_{x | \theta k}^{g M g^\dag} = p_{f(x)|\theta k}^M \, .
$$

Again, we consider the three summands in Eq.~\eqref{eq:merit}
separately. For the first term, observe that $h(p_{0|\theta
k}^M)=h(p_{1|\theta k}^M)$.
\begin{align*}
\H(X | \Theta K)_{g M g^\dag} &= \sum_{\theta k} p_{\theta k}^{g M
g^\dag}  h\left(p_{0|\theta k} ^{g M g^\dag}\right) = \sum_{\theta
k} p_{\theta k}^{M}  h\left(p_{0|\theta k}^M\right) = \H(X | \Theta
K)_{M}\, .
\end{align*}
To analyze the second term, note that we can write
$$
p_{x \theta k}^{g M g^\dagger} = p_{f(x) \theta k}^{ M} \, ,
$$
and for depolarizing noise $\mN\left(U \rho U^\dagger\right) = U
\mN(\rho) U^\dagger$, in addition the von Neumann entropy itself is
invariant under unitary operations $\H(g \mN(\rho) g^\dag) =
\H(\mN(\rho))$. Putting everything together, we obtain
\begin{align*}
\sum_{x \theta k} p_{x \theta k}^\gMg \H\left( \mN\left(\rho_{x
\theta}^{k,\gMg}\right) \right) = \sum_{x \theta k} p_{x \theta k}^M \H\left( \mN\left(\rho_{x
\theta}^{k,M}\right) \right).
\end{align*}
By a similar argument, we derive the equality for the third term
\begin{align*}
\sum_{\theta k} p_{\theta k}^\gMg \H\left( \sum_x p_{x|\theta k}^\gMg \mN\left(\rho_{x
\theta}^{k,\gMg}\right) \right) = \sum_{\theta k} p_{\theta k}^M \H\left( \sum_x
p_{x|\theta k}^M \mN\left(\rho_{x \theta}^{k,M}\right) \right).
\end{align*}
\end{proof}

\begin{claim}\label{maximumInvariant}
  Let $G = \{\id,X,Z,XZ\}$.  There exists a measurement operator $F$
  such that the minimum of $\B(M)$ over all measurements $M$ is
  achieved by a measurement proportional to $\{gFg^\dagger \mid g \in
  G\}$.
\end{claim}
\begin{proof}
Let $M= \{F_k\}$ be a measurement. Let $K = |M|$ be the number of measurement operators.
Clearly, $\hat{M} = \{\hat{F}_{g,k}\}$
with
$$
\hat{F}_{g,k} = \frac{1}{2} g F_k g^\dagger \, ,
$$
is also a quantum measurement since $\sum_{g,k} \hat{F}_{g,k}^\dagger\hat{F}_{g,k} = \id$.
It follows from Claims~\ref{convexity} and~\ref{invarianceOfBunderG} that $\B(M) = \B(\hat{M})$.
Define operators
$$
N_{g,k} = \frac{1}{\sqrt{2\Tr(F_k^\dagger F_k)}}g F_k g^\dagger \, .
$$
Note that
$$
\sum_{g \in G} N_{g,k} = \frac{1}{\sqrt{2\Tr(F_k^\dagger F_k)}} \sum_{u,v \in \01}X^uZ^v
F_k^\dagger F_k Z^v X^u = \id \, .
$$
(see for example Hayashi~\cite{hayashi}).
Hence $M_k = \{N_{g,k}\}$ is a valid quantum measurement.
Now, note that $\hat{M}$ can be obtained from $M_1,\ldots,M_K$ by averaging.
Hence, by Claim~\ref{convexity} we have
$$
\B(M)= \B(\hat{M}) \geq \min_k \B(M_k) \, .
$$
Let $M^*$ be the optimal measurement.  Clearly, $m = \B(M^*) \geq
\min_k \B(M^*_k) \geq m$ by the above and
Claim~\ref{invarianceOfBunderG} from which the present claim
follows.
\end{proof}

Finally, we note that we can restrict ourselves to optimizing over
positive semi-definite (and hence Hermitian) matrices only.
\begin{claim}
Let $F$ be a measurement operator and $M^F = \set{g F g^\dag | g \in
  G}$ the associated measurement. Then there exists a Hermitian
operator $\hat{F}$ such that $\B(M^F)=\B(M^{\hat{F}})$.
\end{claim}
\begin{proof}
Let $F^\dagger = \hat{F}U$ be the polar decomposition of
$F^\dagger$, where $\hat{F}$ is positive semi-definite and $U$ is
unitary~\cite[Corollary 7.3.3]{horn&johnson:ma}.  Evidently, since
the trace is cyclic, all probabilities remain the same. Using the
invariance of the von Neumann entropy and the depolarizing quantum
operation under unitaries, the claim follows.
\end{proof}

Note that Claim~\ref{maximumInvariant} also gives us that we have at
most 4 measurement operators. Wlog, we take the measurement outcomes
to be labeled $1,2,3,4$ and measurement operators $F_1 = F, F_2=XFX,
F_3=ZFZ, F_4=XZFZX$. Our final observation is the following easy
claim.
\begin{claim}\label{projector}
  For any linear operator $F$ on Hilbert space $\hil$ and any state
  $\ket{\phi} \in \hil$ such that $F \ket{\phi} \neq 0$, it holds that
  the operator $P \assign \frac{F \proj{\phi} F^\dag}{\Tr(F
    \proj{\phi} F^\dag )}$ is a projector with $\rank(P)=1$.
\end{claim}
\begin{proof}
Notice that $\proj{\phi} F^\dag F \proj{\phi} =
\Tr(F \proj{\phi} F^\dag) \proj{\phi}$. Thus
\begin{align*}
 PP=\frac{F \proj{\phi} F^\dag F \proj{\phi} F^\dag}{\Tr(F
  \proj{\phi} F^\dag )^2} = \frac{F \proj{\phi} F^\dag}{\Tr(F
  \proj{\phi} F^\dag )} = P \, .
\end{align*}
As $F \ket{\phi} \neq 0$ we have that $\rank(F \proj{\phi}
F^\dag)=1$.
\end{proof}

Exploiting our observations, we can considerably simplify the
expression $B(M)$ to be minimized:

\begin{lemma}
Let $\B(M)$ be defined as above. Then
$$
\min_M \B(M) = \min_F  C(F),
$$
where the minimization is taken over Hermitian operators $F \in \Complex^{2\times 2}$ and $C(F)$ is defined
as
\begin{equation}
C(F)= \frac12 \left( h \left(2\,  \Tr\left(F \rho_{0 +}
F\right) \right) + h\left(2\, \Tr\left(F \rho_{0
  \times} F\right) \right) \right) + h\left(\frac{1+r}{2}\right) -
\H\left( \mN \left( 2 F^2\right) \right).
\label{def:C}
\end{equation}
\end{lemma}
\begin{proof}
First of all, note that
$$
p_{\theta k} = p_{0 \theta k} + p_{1 \theta k} = \frac{1}{4}
\Tr\left( F_k \left(\rho_{0 \theta}+\rho_{1 \theta}\right) F_k
\right) = \frac{1}{4} \Tr( F^2) \, ,
$$
which is independent of $k$. Thus we have
$$
\frac12=p_{\theta} = \sum_{k=1}^4 p_{\theta k} = \Tr( F^2) \, ,
$$
and hence $p_{\theta k} = \frac{1}{8}$. Furthermore, as in the proof
of Claim~\ref{invarianceOfBunderG}, there exists a bijection $f: \01
\rightarrow \01$ such that
$$
p_{x | \theta k} = \frac{p_{x \theta k}}{p_{\theta k}} = \frac{ \Tr\left(F_k
      \rho_{x \theta} F_k \right)/4}{1/8} = 2\, \Tr \left( F_k \rho_{x \theta} F_k \right)
= 2\, \Tr\left(F \rho_{f(x) \theta}F \right) \, .
$$
Note again that $h(p_{0|\theta k}^M)=h(p_{1|\theta k}^M)$.
We then obtain for the first term
\begin{align*}
\H(X | \Theta K) &= \sum_{\theta k} p_{\theta k} h(p_{0|\theta k}) \\
&= \sum_{\theta k} \frac{1}{8} h( 2\, \Tr( F_k \rho_{0 \theta} F_k) ) \\
&= \frac12 \left( h (2\,  \Tr(F \rho_{0 +} F) ) + h(2\, \Tr(F
\rho_{0 \times} F) ) \right) \, .
\end{align*}

For the second term, we need to evaluate $\H( \mN(\rho_{x \theta}^k
))$. It follows from Claim~\ref{projector} that if $p_{x \theta k} >
0$, the normalized post-measurement state $\rho_{x \theta}^k$ has
eigenvalues $0$ and $1$. Applying the depolarizing quantum operation
to such rank 1 state gives an entropy $\H( \mN(\rho_{x \theta}^k)) =
h((1+r)/2)$, independent of the state. Thus the second term becomes
\begin{align*}
\sum_{x \theta k} p_{x \theta k} \H\left( \mN\left(\rho_{x \theta}^k \right)\right) &=
\sum_{x \theta k} p_{x \theta k} h\left(\frac{1+r}{2}\right) =h\left(\frac{1+r}{2}\right).
\end{align*}

For the third term, we use that for $0 \leq \alpha \leq 1$, it holds
that $\mN (\alpha \rho + (1-\alpha) \sigma)= \alpha \mN(\rho) +
(1-\alpha) \mN(\sigma)$. Hence,
\begin{align*}
p_{0 | \theta k} \mN\left(\rho_{0  \theta}^k\right) + p_{1 | \theta k}
\mN\left(\rho_{1 \theta}^k\right) &=
\mN \left( p_{0 | \theta k} \rho_{0 \theta}^k + p_{1 | \theta k} \rho_{1
  \theta}^k\right)\\
&=\mN\left( 2\, \Tr\left(F_k \rho_{0 \theta} F_k^\dag\right)
\frac{F_k \rho_{0 \theta} F_k^\dag}{\Tr\left(F_k \rho_{0 \theta}
F_k^\dag\right)} +  2\, \Tr\left(F_k \rho_{1 \theta} F_k^\dag\right)
\frac{F_k \rho_{1 \theta} F_k^\dag}{\Tr\left(F_k \rho_{1 \theta}
F_k^\dag\right)}
\right)\\
&=\mN \left(2 F_k (\rho_{0 \theta} + \rho_{1 \theta}) F_k^\dag \right)\\
&=U_k \mN\left(2 F^2\right) U_k^{\dagger},
\end{align*}
where $U_k \in G$.
The third term then yields
$$
\sum_{\theta k} p_{\theta k} \H\left( \sum_x p_{x|\theta k}
\mN\left(\rho_{x \theta}^k \right)\right) = \H\left( \mN ( 2 F^2)
\right) \, .
$$
These arguments prove the Lemma.
\cancel{ To summarize, our optimization problem can be simplified to
\begin{eqnarray*}
\min_M \B(M) & \equiv & \min_M \H(X|\Theta K E)_M = \min_F
\H(X|\Theta K
E)_{M^F}\\
&=& \min_F \C(F) \, ,
\end{eqnarray*}
for
\begin{equation} \label{def:C}
\C(F) \assign \frac12 \left( h \left(2  \Tr\left(F \rho_{0 +} F\right) \right) + h\left(2 \Tr\left(F \rho_{0
  \times} F\right) \right) \right) + h\left(\frac{1+r}{2}\right) -
\H\left( \mN \left( 2 F F\right) \right) \, ,
\end{equation}
as promised.}
\end{proof}

\subsection{F is Diagonal in the Computational Basis}

Now that we have simplified our problem considerably, we are ready
to perform the actual optimization. We first show that we can take
$F$ to be diagonal in the computational (or Hadamard) basis.
\begin{claim}\label{ignoreY}
Let $F \in \Complex^{2 \times 2}$ be the Hermitian operator that
minimizes $\C(F)$ as defined by Eq.~\eqref{def:C}.
Then $F = \alpha\proj{\phi}+\beta\, (\id-\proj{\phi})$ for some $\alpha, \beta \in \Real$
and pure state $\ket{\phi}$ lying in the XZ plane of the Bloch sphere. (i.e. $\Tr(FY) = 0$).
\end{claim}
\begin{proof}
Since $F$ is a Hermitian on a 2-dimensional space, we may express $F$ as
$$
F = \alpha \proj{\phi} + \beta \proj{\phi^\perp} \, ,
$$
for some state $\ket{\phi}$ and real numbers $\alpha,\beta$. We first of all note that from
$\sum_{k} F_k F_k = \id$, we obtain that
\begin{equation*}
\sum_k \Tr(F_k F_k) =\sum_{g \in\{\id,X,Z,XZ\}} \Tr(gFgg^\dagger Fg^\dagger)= 4\, \Tr(F^2) = \Tr(\id) = 2 \, ,
\end{equation*}
and hence $\Tr(F^2) = \alpha^2 + \beta^2 = 1/2$. Furthermore, using that
$\proj{\phi} + \proj{\phi^\perp} =\id$ gives
\begin{equation}\label{Fdef}
F = \alpha\proj{\phi}+\beta\, (\id- \proj{\phi}) \, ,
\end{equation}
with $\beta = \sqrt{1/2-\alpha^2}$.
Hence without loss of generality, we can consider $0 \leq \alpha \leq 1/\sqrt{2}$.
The eigenvalues of $2 F^2$ are
$2\alpha^2$ and $1 - 2 \alpha^2$. Hence, the third term of $\C(F)$ becomes
$\H(\mN(2 F^2)) = h(2 r \alpha^2 + (1-r)/2)$ which does not depend on
$\ket{\phi}$. We want to minimize
\begin{equation} \label{eq:mintarget}
\min_F \frac12 ( h (2\,  \Tr(F \rho_{0 +} F) ) + h(2\, \Tr(F \rho_{0
  \times} F) ) ) +
h((1+r)/2) - h(2 r \alpha^2 + (1-r)/2) \, .
\end{equation}
We first parametrize the state $\ket{\phi}$ in terms of its Bloch vector
$$
\proj{\phi} = \frac{\id + \hat{x} X + \hat{y} Y + \hat{z} Z}{2} \, .
$$
Since $\ket{\phi}$ is pure we can write $\hat{y} = \sqrt{1 - \hat{x}^2 - \hat{z}^2}$.
Note that we may wlog assume that $0 \leq \hat{x}, \hat{z} \leq 1$,
since the remaining three
measurement operators are given by $XFX$, $ZFZ$, and $XZFZX$.
A small calculation shows that for the encoded bit $x \in \01$
$$
2\,\Tr\left(F\rho_{x+}F\right) = \frac{1}{2}\left(1 + (-1)^x (4 \alpha^2
  - 1)\hat{z}\right) \, ,
$$
and similarly
$$
2\,\Tr\left(F\rho_{x\times}F\right) = \frac{1}{2}\left(1 + (-1)^x (4
  \alpha^2 - 1)\hat{x}\right) \, .
$$
Our goal is to show that for every $0 \leq \alpha \leq 1/\sqrt{2}$,
the function
$$
f(\hat{z}) \assign h(2\, \Tr(F\rho_{x+}F))
$$
is non-increasing
on the interval $0 \leq \hat{z} \leq 1$.  First of all, note
that $f(\hat{z}) = 1$ for $\alpha=1/2$. We now consider the case of
$\alpha \neq 1/2$.
A simple computation (using Mathematica) shows that when differentiating
$f$ with respect to $\hat{z}$ we obtain
\begin{align*}
f'(\hat{z}) &=
\frac{\partial}{\partial \hat{x}} f(\hat{z}) =
\frac{1}{2}(1 - 4 \alpha^2) \log\left(\frac{2}{1 + \hat{z} - 4 \alpha^2 \hat{z}} - 1\right),\\
f''(\hat{z}) &=
\frac{\partial^2}{\partial \hat{x}} f(\hat{z}) =
\frac{(1 - 4 \alpha^2)^2}{\ln 2 (\hat{z}^2(1 - 4 \alpha^2)^2 - 1)}
\end{align*}
Hence the function has one maximum at $\hat{z}=0$ with $f(0)=1$.
Since $0\leq \alpha \leq 1/\sqrt{2}$ and $\alpha \neq 1/2$ we also have that
$(1-4\alpha^2)^2 \leq 1$ and hence $f''(\hat{z}) \leq 0$ everywhere
and $f$ is concave (though not strictly concave).
Thus $f(\hat{z})$ is decreasing with $\hat{z}$.

Since we have $\hat{x}^2 + \hat{z}^2 + \hat{y}^2 = 1$ we can thus
conclude that in order to minimize $\C(F)$, we want to choose
$\hat{x}$ and $\hat{z}$ as large as possible and thus let $\hat{y} =
0$ from which the claim follows.
\end{proof}

We can immediately extend this analysis to find
\begin{claim}
Let $F$ be the operator that minimizes $\C(F)$, and write $F$ as in Eq.~\ref{Fdef}.
Then
$$
\ket{\phi} = g \ket{0} \, ,
$$
for some $g \in \{\id,X,Z,XZ\}$.
\end{claim}
\begin{proof}
By Claim~\ref{ignoreY}, we can rewrite our optimization problem as
\begin{sdp}{minimize}{$(f(\hat{x}) + f(\hat{z}))/2 + h((1+r)/2) - h(2 r \alpha^2 + (1-r)/2)$}
&$\hat{x}^2+\hat{z}^2 = 1$\\
&$0\leq \hat{x} \leq 1$\\
&$0 \leq \hat{z} \leq 1 \,$.
\end{sdp}
By using Lagrange multipliers we can see that for an extreme point
we must have either $\hat{x}=\hat{z} = 1/\sqrt{2}$ or $\hat{x} = 0,
\hat{z}=1$ or $\hat{z} = 0, \hat{x}=1$. From the definition of $f$ above we can see that to minimize the expression, we want to choose
the latter, from which the claim follows.
\end{proof}

\subsection{Optimality of the Trivial Strategies}
We have shown that without loss of generality $F$ is diagonal
in the computational basis. Hence, we have only a single parameter left in
our optimization problem. We must optimize over all operators $F$
of the form
$$
F = \alpha \proj{\phi} + \sqrt{1/2 - \alpha^2} \proj{\phi^\perp} \, ,
$$
where we may take $\ket{\phi}$ to be $\ket{0}$ or $\ket{1}$. Our aim is to show that
either $F$ is the identity, or $F = \proj{\phi}$ depending on the value of $r$.
\begin{claim}
Let $F$ be the operator that minimizes $\C(F)$, and let
$r_0 \assign 2 h^{-1}\left(\frac{1}{2}\right) - 1$.
Then $F = c\, \id$ (for some $c \in \Real$) for $r \geq r_0$, and
$F = \proj{\phi}$ for
$r < r_0$, where
$$
\ket{\phi} = g \ket{0} \, ,
$$
for some $g \in \{\id,X,Z,XZ\}$.
\end{claim}
\begin{proof}
We can plug $x=0$ and $z=1$ in the expressions in the proof of our previous claim.
Thus our goal is to minimize
$$
t(r,\alpha) \assign \frac{1}{2}\left(1 + g(1,\alpha)\right) +
h\left(\frac{1+r}{2}\right) - g(r,\alpha) \, ,
$$
where
$$
g(r,\alpha) \assign h\left(\frac{1+r}{2}-2 \alpha^2 r\right) \, .
$$
Differentiating $g$ with respect to $\alpha$ gives us
$$
\frac{\partial}{\partial \alpha} g(r,\alpha) = 4 \alpha r
\left(\log\left(\frac{1+r}{2} - 2 \alpha^2 r\right) -
  \log\left(\frac{1-r}{2} + 2 \alpha^2 r\right)\right) \, ,
$$
with which we can easily differentiate $t$ with respect to $\alpha$ as
$$
\frac{\partial}{\partial \alpha} t(r,\alpha) = \frac{1}{2}
\frac{\partial}{\partial \alpha} g(1,\alpha) -
\frac{\partial}{\partial \alpha} g(r,\alpha) \, .
$$
We can calculate
$$
\lim_{\alpha\rightarrow 0} \frac{\partial}{\partial \alpha} t(r,\alpha) = 0
$$
and
$$
\frac{\partial}{\partial \alpha} t(r,1/2) = 0 \, .
$$
We thus have two extremal points. By computing the second derivative which is equal to $8(2r^2-1)/\ln 2$ at the point $\alpha=1/2$,
we can see that as $r$ grows from 0 to 1, the second extreme point switches from a maximum to a minimum at $r = 1/\sqrt{2}$.
Our goal is thus to determine for which $r$ we have
$$
t(r,0) \leq t(r,1/2) \, .
$$
Note that shortly after the transition point $r=1/\sqrt{2}$, we do obtain two additional maxima, but since we
are interested in finding the minimum they do not contribute to our analysis.
By plugging in the definition for $t$ from above, we have that $t(r,0) \leq t(r,1/2)$
iff
$$
\frac{1}{2} \leq h\left(\frac{1+r}{2}\right) \, ,
$$
or in other words iff
$$
2 h^{-1}\left(\frac{1}{2}\right) - 1 \leq r \, ,
$$
as promised.
\end{proof}

We conclude that Bob's optimal strategy, --the one which minimizes $H(X|\Theta K E)$--,
is an extremal strategy, that is, he either measures his qubit in the computational
basis, or he stores the qubit as is. This is the content of Theorem~\ref{thm:depolarize}.
We believe that a similar analysis can be done for the dephasing
quantum operation, by first symmetrizing
the noise by applying a rotation over $\pi/4$ to the input states.

\end{document}